\documentclass[runningheads]{llncs}
\usepackage[utf8]{inputenc}
\usepackage{amssymb}
\usepackage{listings}
\usepackage{amsfonts}
\usepackage{float}
\usepackage{amsmath,latexsym}
\usepackage{graphicx}
\usepackage{fancyvrb}
\usepackage{authblk}
\usepackage{paralist}
\usepackage{makecell}
\usepackage{comment}
\usepackage{cite}
\DeclareMathOperator{\lcm}{lcm}
\usepackage[table,xcdraw]{xcolor}
\newif\ifanonymous
\anonymousfalse  % quick toggle to anonymize
\usepackage{xcolor}
\usepackage{tikz-cd}

\usepackage{xcolor}
\definecolor{linkcolor}{rgb}{0.65,0,0}
\definecolor{citecolor}{rgb}{0,0.4,0}
\definecolor{urlcolor}{rgb}{0,0,0.65}
\usepackage[colorlinks=true, linkcolor=linkcolor, urlcolor=urlcolor, citecolor=citecolor]{hyperref}
\definecolor{darkblue}{RGB}{0,0,160}
\definecolor{darkdarkred}{RGB}{180,0,0}
\definecolor{darkgreen}{RGB}{0,140,0}
\newcommand{\FF}{\mathbb{F}}
\newcommand{\FFt}{\mathbb{F}_2}
\newcommand{\FFq}{\mathbb{F}_q}

\newcommand{\K}{\mathbb{K}}

\newcommand{\vs}{\mathbf{s}}

\newcommand{\vc}{\mathbf{c}}
\newcommand{\ve}{\mathbf{e}}

\newcommand{\vv}{\mathbf{v}}

\newcommand{\vx}{\mathbf{x}}
\newcommand{\vy}{\mathbf{y}}

\newcommand{\vz}{\mathbf{z}}
\newcommand{\vH}{\mathbf{H}}

\newcommand{\htop}{{\mathrm{top}}}

\newtheorem{modeling}{Modeling}
\newtheorem{notation}{Notation}

\newcommand{\Cf}{\mathbf{C}_f}
\newcommand{\HH}{\mathbf{H}}

\newcommand{\CC}{\mathcal{C}}
\newcommand{\OO}{\mathcal{O}}
\newcommand{\GG}{\mathcal{G}}

\newcommand{\supp}{\mathsf{supp}}

\newcommand{\rk}{\mathsf{rk}}

\newcommand{\wt}{\mathsf{wt}}
\newcommand{\lm}{\mathsf{lm}}

\newcommand{\dreg}[1]{d_{\mathrm{reg}}(#1)}
\newcommand{\pr}{{\mathbb{P}}}
\newcommand{\ord}{\mathsf{ord}}

\begin{document}

 \title{Quadratic Modelings of Syndrome Decoding} %Multivariate Modelings of Syndrome Decoding 

  \author{Alessio Caminata \inst{1} \and
  Ryann Cartor \inst{2}\and
  Alessio Meneghetti \inst{3}\and 
  Rocco Mora \inst{4} \and
    Alex Pellegrini \inst{5}}
% % %
  \authorrunning{A. Caminata et al.}

 \institute{Universit\`a di Genova
 \and
 Clemson University
 \and
 Universit\`a di Trento 
 \and CISPA Helmholtz Center for Information Security 
 \and Eindhoven University of Technology
 }

    \maketitle
    \begin{abstract} 
        This paper presents enhanced reductions of the bounded-weight and exact-weight Syndrome Decoding Problem (SDP) to a system of quadratic equations. Over $\FFt$, we improve on a previous work and study the degree of regularity of the modeling of the exact weight SDP.  Additionally, we introduce a novel technique that transforms SDP instances over $\FF_q$ into systems of polynomial equations and thoroughly investigate the dimension of their varieties. Experimental results are provided to evaluate the complexity of solving SDP instances using our models through Gr\"obner bases techniques.
        
    \end{abstract}

    \section{Introduction}\label{sec:intro}

As widespread quantum computing becomes closer to reality, accurate cryptanalysis of post-quantum cryptosystems is of the utmost importance. Code-based cryptography is one of the main areas of focus in the search for quantum-secure cryptosystems.
This is well represented by the NIST Post-Quantum Standardization Process, where as many as three finalists, namely Classic McEliece \cite{bernstein2017classic} 
 (an IND-CCA2 secure variation of McEliece's very first code-based scheme \cite{mceliece1978public}), HQC \cite{melchor2018hamming} and BIKE \cite{aragon2022bike}, belong to this family. Similarly, NIST's additional call for digital signatures has numerous proposals that make use of linear codes. Many of the proposed schemes are based on the hardness of (sometimes structured variants of) the syndrome decoding problem. 

The parameters of many code-based schemes are carefully chosen to align with the latest advancements with respect to this computational problem. Despite decades of intensive research in this direction, all the algorithms developed so far exhibit exponential complexity. This is not surprising, since the problem has been shown to be NP-hard \cite{berlekamp1978inherent}. In particular, after more than 60 years of investigation since the groundbreaking paper of Prange \cite{DBLP:journals/tit/Prange62}, the reduction in the exponent for most parameters of interest has been minimal \cite{stern1989method, D89, finiasz2009security, bernstein2011smaller, may2011decoding, becker2012decoding, may2015computing, both2018decoding}. All the works mentioned fall into the family of Information Set Decoding (ISD) algorithms, whose basic observation is that it is easier to guess error-free positions, and guessing enough of them is sufficient to decode. 
This resistance to ISD algorithms makes the syndrome decoding problem a reliable foundation for code-based cryptosystems.

     To comprehensively assess security, it is imperative to consider attacks stemming from various other realms of post-quantum cryptography. For instance, attacks typically associated with multivariate or lattice-based schemes should also be taken into account for code-based schemes, when applicable. A remarkable example is offered by dual attacks, originally introduced in lattice-based cryptography, where, however, they have been strongly questioned. In contrast, their code-based counterpart \cite{carrier2022statistical, carrier2024reduction} has recently outperformed ISD techniques for a non-negligible regime of parameters, by reducing the decoding problem to the closely related Learning Parity with Noise problem.
     Concerning polynomial system solving strategies, another notable illustration of this is the algebraic MinRank attack, which broke the rank-metric code-based schemes RQC and Rollo \cite{bardet2020algebraic, DBLP:conf/asiacrypt/BardetBCGPSTV20} and now represents the state-of-the-art for MinRank cryptanalysis, beating combinatorial approaches.

    In the Hamming metric, a reduction that transforms an instance of the syndrome decoding problem into a system of quadratic equations over $\mathbb{F}_2$ was introduced in \cite{2021/meneghetti}.
    The most expensive step of the transformation, in terms of numbers of new variables and new equations introduced, is the so-called \textit{Hamming-weight computation encoding}. Indeed, for a binary linear code of length $n$, the procedure dominates the overall complexity of the reduction with a complexity of $\mathcal{O}(n\log_2(n)^2)$.%, where $\ell = \lfloor \log_2(n) \rfloor + 1$. 

 Despite the considerable theoretical interest in this transformation, the latter is too inefficient to be of practical interest in solving the syndrome decoding problem. Thus, the problem of improving the reduction in order to obtain a more effectively solvable system remains open. Moreover, \cite{2021/meneghetti} covers only the binary case, leaving unanswered the challenge of modeling through algebraic equations the decoding problem for codes defined over finite fields with more than two elements.

%%%%%%%%%%%%%%%%%%%%%%%%%%%%%%%%	
\paragraph{Our contribution.}

In this work, we improve on the reduction presented in \cite{2021/meneghetti} by a factor of \(\log_2(n)\), thereby reducing the number of introduced variables and equations and achieving an overall reduction cost of \(\mathcal{O}(n\log_2(n))\). This improvement is achieved by leveraging the recursive structure of the equations generated by the Hamming-weight computation encoding and by transforming the equations similarly to the reduction procedure in Buchberger's algorithm \cite{1965/buchberger} for Gröbner basis computation. When considering a version of the syndrome decoding problem that requires an error vector with a specified Hamming weight, we derive a further improved modeling, for which we study the degree of regularity.

As a second contribution, we present a novel approach that transforms an instance of the syndrome decoding problem over \(\mathbb{F}_{q}\) for \(q \geq 2\) into a system of polynomial equations. This significantly broadens the applicability of our methods to a wider range of code-based cryptosystems. A common feature of our algebraic modelings is that if the decoding problem admits multiple solutions, the Gröbner basis naturally determines all of them.

We also provide theoretical and experimental data to analyze the complexity of solving syndrome decoding instances using our modelings, demonstrating that, at least for small parameters, our new strategy is practical and successful. Software (MAGMA scripts) supporting this work can be found \href{https://github.com/rexos/phd-cryptography-code/tree/main/modelings}{here}.

\paragraph{Structure of the paper.}
The next section recalls the background and notions necessary for this work. In Section~\ref{sec:mps}, we review the reduction described in \cite{2021/meneghetti} from the syndrome decoding problem to that of finding the zeroes of a set of polynomials. In Section~\ref{sec:EWM}, we describe two modelings that improve upon \cite{2021/meneghetti}. We study the degree of regularity of the modeling for the exact weight syndrome decoding problem, along with experimental results, in Section~\ref{sec:complexity-analysis}. Finally, in Section~\ref{sec:Fq}, we present a novel modeling of the syndrome decoding problem over $\mathbb{F}_{q}$ with $q \geq 2$, for which we provide a theoretical study of the variety and experimental analysis of the solving complexity with Gr\"obner bases techniques.

%%%%%%%%%%%%%%%%%%%%%%%%%%%%%%%%
%%%%%%%%%%%%%%%%%%%%%%%%%%%%%%%%
    \section{Preliminaries} \label{sec:prelim}
    
This paper investigates the reduction of the Syndrome Decoding Problem (SDP) into a Polynomial System Solving Problem (PoSSo). In this section, we briefly recall the definitions of both problems, as well as the notions of solving degree and degree of regularity, which are commonly used to estimate the computational complexity of the PoSSo problem.
  
 %%%%%%%%%%%%%%%%%%%%%%%%%%%%%%%%
    \subsection{The Syndrome Decoding Problem}
    An $[n,k]$-linear code $\mathcal{C}$ is a $k$-dimensional subspace of $\FF_q^n$. We call $n$ the length of the code, and $k$ its dimension. An element $\mathbf{x}\in\FF_q^n$ is called a codeword if $\mathbf{x}\in\mathcal{C}$. The number of nonzero entries in $\mathbf{x}$ is called the Hamming weight of $\mathbf{x}$ and we denote it as $\wt(\mathbf{x})$. Given a code $\mathcal{C}$ we define a parity check matrix of $\mathcal{C}$ as $\mathbf{H}\in\FF_q^{(n-k)\times n}$ such that the right kernel of $\mathbf{H}$ is the code $\mathcal{C}$. The subspace spanned by the rows of $\HH$ is called the dual code of $\mathcal{C}$. 
     Many code-based cryptosystems rely on the hardness of solving the Syndrome Decoding Problem (SDP), see Problems~\ref{BSDP} and~\ref{EWSDP} described below.

    \begin{problem}[SDP: Syndrome Decoding Problem]\label{BSDP}
    Given integers $n,k,t$ such that $k\leq n$ and $t\leq n$, an instance of the problem SD$(\HH,\mathbf{s},t)$ consists of a parity check matrix $\mathbf{H}\in\FF_q^{(n-k)\times n}$ and a vector $\mathbf{s}\in\FF_q^{n-k}$ (called the syndrome). A solution to the problem is a vector $\mathbf{e}\in \mathbb{F}_q^n$ such that $\mathbf{He}^\top=\mathbf{s}^\top$ and $\wt(\mathbf{e})\leq t$.
    \end{problem} 

    \noindent In later sections, we will also refer to Problem~\ref{BSDP} as the ``Bounded Syndrome Decoding" Problem. We will also consider the following variant of SDP. 

    \begin{problem}[ESDP: Exact Weight Syndrome Decoding Problem]\label{EWSDP}
     Given integers $n,k,t$ such that $k\leq n$ and $t\leq n$, an instance of the problem ESD$(\HH,\mathbf{s},t)$ consists of a parity check matrix $\mathbf{H}\in\FF_q^{(n-k)\times n}$ and a vector $\mathbf{s}\in\FF_q^{n-k}$ (called the syndrome). A solution to the problem is a vector $\mathbf{e}\in \mathbb{F}_q^n$ such that $\mathbf{He}^\top=\mathbf{s}^\top$ and $\wt(\mathbf{e})= t$.
    \end{problem}

    Additionally, a close variant of the Syndrome Decoding Problem is the \textit{Codeword Finding Problem}, where the syndrome $\vs$ is the zero vector ${\mathbf{0}}$. Since the null vector is always a solution of the parity-check equations $\mathbf{He}^\top=\mathbf{0}^\top$, a nonzero $\ve$ of weight at most (or exactly) $t$ is sought. The name of the problem refers to the fact that any element in the right kernel of $\mathbf{H}$ belongs to the code $\mathcal{C}$ having $\HH$ as parity-check matrix. We will later need to distinguish this variant in the analysis of one of our modelings.

In addition to length and dimension, a fundamental notion in coding theory and consequently in code-based cryptography is the minimum distance $d$ of an $\FF_q$-linear code, i.e. the Hamming weight of the smallest nonzero codeword in the code. Such a quantity is strictly related to the number of solutions to the syndrome decoding problem.

Knowing the expected number of solutions from given parameters is extremely important in cryptography, in order to assess the security correctly.
It is guaranteed that the problem does not admit more than one solution as long as the number of errors is upper bounded by $\frac{d-1}{2}$. However, in practice, much better can be done for randomly generated codes. Indeed, it turns out that random codes achieve the so-called Gilbert-Varshamov (GV) distance $d_{GV}$, defined as the largest integer such that
\[
\sum_{i=0}^{d_{GV}-1} \binom{n}{i}(q-1)^i \le q^{n-k}.
\]

It can be shown that, as long as the number of errors is below the Gilbert-Varshamov distance, the Syndrome Decoding problem \textit{typically} has a unique solution. Moreover, the instances where the number of errors attains the GV distance are those supposed to be the most difficult.

%%%%%%%%%%%%%%%%%%%%%%%%%%%%%%%%
    \subsection{The Polynomial System Solving Problem}
The Polynomial System Solving Problem (PoSSo) is the following. We define it over a finite field $\FF_q$, athough it can be more generally considered over any field.
 \begin{problem}[PoSSo: Polynomial System Solving]\label{PoSSo}
 Given integers $N,r\geq2$, an instance of the PoSSo problem consists of a system of  polynomials $\mathcal{F}=\{f_1,\dots,f_r\}$ in $R=\FF_q[x_1,\dots,x_N]$ with $N$ variables and coefficients in $\FF_q$. A solution to the problem is a vector $\mathbf{a}\in\FF_q^N$ such that $f_1(\mathbf{a})=\cdots=f_r(\mathbf{a})=0$.
 \end{problem}

\begin{remark}A special case of PoSSo when $\deg(f_i)=2$ for $1\leq i\leq r$  is called MQ (Multivariate Quadratic) and is the basis for multivaritate cryptography.
\end{remark}
The following outlines a standard strategy for finding the solutions of a polynomial system $\mathcal{F}$ by means of Gr\"obner bases. 
\begin{compactenum}
\item Find a degree reverse lexicographic ($\mathsf{degrevlex}$) Gr\"obner basis of the ideal $\langle\mathcal{F}\rangle$;
\item Convert the obtained $\mathsf{degrevlex}$ Gr\"obner basis into a lexicographic ($\mathsf{lex}$) Gr\"obner basis, where the solutions of the system can be easily read from the ideal in this form.
\end{compactenum}
The second step can be done by FGLM \cite{FGLM93}, or a similar algorithm, whose complexity depends on the degree of the ideal. This is usually faster than the first step, especially when the system $\mathcal{F}$ has few solutions. Therefore, we focus on the first step. 

The fastest known algorithms to compute a $\mathsf{degrevlex}$ Gr\"obner basis are the linear algebra based algorithms such as F4 \cite{faugereF4}, F5 \cite{F5paper}, or XL \cite{XL00}. These transform the problem of computing a Gr\"obner basis into one or more instances of Gaussian elimination of the  Macaulay matrices. The complexity of these algorithms is dominated by the Gaussian elimination on the largest Macaulay matrix encountered during the process. The size of a Macaulay matrix depends on the degrees of the input polynomials $f_1,\dots,f_r$, on the number of variables $N$, and on a degree $d$. In a nutshell, the \emph{Macaulay matrix} $M_{\leq d}$ of degree $d$ of $\mathcal{F}$ has columns indexed by the monic monomials of degree $\leq d$, sorted in decreasing order from left to right (with respect to the chosen $\mathsf{degrevlex}$ term order). The rows of $M_{\leq d}$ are indexed by the polynomials $m_{i,j}f_j$, where $m_{i,j}$ is a monic monomial such that $\deg(m_{i,j}f_j)\leq d$.  The entry $(i,j)$ of $M_{\leq d}$ is the coefficient of the monomial of column $j$ in the polynomial corresponding to the $i$-th row.

The \emph{solving degree} of $\mathcal{F}$ is defined as the least degree $d$ such that Gaussian elimination on the Macaulay matrix $M_{\leq d}$ produces a $\mathsf{degrevlex}$ Gr\"obner basis of $\mathcal{F}$. We denote the solving degree of $\mathcal{F}$ by $d_{\mathrm{sol}}(\mathcal{F})$. We have to compute Macaulay matrices up to degree $d_{\mathrm{sol}}=d_{\mathrm{sol}}(\mathcal{F})$, and the largest one we encounter has $a=\sum_{i=1}^r{{N+d_{\mathrm{sol}}-d_i}\choose{d_{\mathrm{sol}}-d_i}}$ many rows and $b={{N+d_{\mathrm{sol}}}\choose{d_{\mathrm{sol}}}}$ many columns, where $d_i=\deg f_i$. Therefore, taking into account the complexity of Gaussian elimination of this matrix,  an upper bound on the complexity of solving the system $\mathcal{F}$ with this method is
\begin{equation}\label{eq:GBcomplexity}
\OO\left({{N+d_{\mathrm{sol}}}\choose{d_{\mathrm{sol}}}}^\omega\right),    
\end{equation}
with $2\leq\omega\leq3$.

\begin{remark}
If $\mathcal{F}$ is not homogeneous, Gaussian elimination on $M_{\leq d}$ may produce a row corresponding to a polynomial $f$ with $\deg f<d$, where the leading term of $f$ was not the leading term of any row in $M_{\leq d}$. Some algorithms, for example $F4$, address this by adding rows for polynomials $mf$ ($\deg(mf)\leq d$) for some monomial $m$ and recomputing the reduced row echelon form.
If no Gr\"obner basis is found in degree $\leq d$, they proceed to higher degrees, potentially enlarging the span of $M_{\leq d}$ and reducing the solving degree. Throughout this paper, we consider only the case where no extra rows are added. Note that the solving degree as defined above is an upper bound on the degree at which algorithms using this variation terminate.
\end{remark}

Since the solving degree of a polynomial system may be difficult to estimate, several invariants related to the solving degree (that are hopefully easier to compute) have been introduced. One of the most important is the \emph{degree of regularity} introduced by Bardet, Faug\`ere, and Salvy \cite{bardet2004complexity}. We briefly recall its definition and connection with the solving degree.

Let $\langle\mathcal{F}^{\mathrm{top}}\rangle=\langle f_1^{\mathrm{top}},\dots,f_r^{\mathrm{top}}\rangle$ be the ideal of the polynomial ring $R$ generated by the homogeneous part of highest degree of the polynomial system $\mathcal{F}$. Assume that $\langle\mathcal{F}^{\mathrm{top}}\rangle_d=R_d$ for $d\gg0$.
The \emph{degree of regularity} of $\mathcal{F}$ is
\begin{equation*}
\dreg{\mathcal{F}}=\min\{d\in\mathbb{N}\mid \langle\mathcal{F}^{\mathrm{top}}\rangle_e=R_e \ \forall e\geq d\}.
\end{equation*} 

The degree of regularity can be read off from the Hilbert series of $\langle\mathcal{F}^{\mathrm{top}}\rangle$.
Let $I$ be a homogeneous ideal of $R$, and let $A=R/I$.
For an integer $d\geq 0$, we denote by $A_d$ the homogeneous component of degree $d$ of $A$.
The function $\mathrm{HF}_A(-):\mathbb{N}\rightarrow\mathbb{N}$, $\mathrm{HF}_A(d)=\dim_{\FF_q}A_d$ is called \emph{Hilbert function} of $A$.

The generating series of $\mathrm{HF}_A$ is called \emph{Hilbert series} of $A$. We denote it by $\mathrm{HS}_A(z)=\sum_{d\in\mathbb{N}}\mathrm{HF}_A(d)z^d$.
 
\begin{remark}\label{rem:polyHS}
Under the assumption that  $\langle\mathcal{F}^{\mathrm{top}}\rangle_d=R_d$ for $d\gg0$,  the Hilbert series of $A=R/\langle\mathcal{F}^{\mathrm{top}}\rangle$ is a polynomial. Then, the degree of regularity of $\mathcal{F}$ is given by $\dreg{\mathcal{F}}=\deg \mathrm{HS}_A(z)+1$ (see \cite[Theorem~12]{2021/caminatagorla}).    
\end{remark}

\noindent Under suitable assumptions, the degree of regularity provides an upper bound for the solving degree \cite{CaminataG23, 2023/salizzoni, Semaev2021651}. Moreover, it is often assumed that the two values are close. Although this occurs in many relevant situations, there are examples where these two invariants can be arbitrarily far apart (see \cite{2021/caminatagorla, 2013/dingschmidt, Bigdeli202175}).  We will see in Section~\ref{sec:dreg-EWM} that the degree of regularity of the system presented in Section~\ref{subsec:f2ESD} seems to yield a much higher value than the solving degree achieved during the Gr\"obner basis algorithm.

%%%%%%%%%%%%%%%%%%%%%%
%%%%%%%%%%%%%%%%%%%%%%
\section{The MPS Modeling}%{The ~\cite{2021/meneghetti} modeling}
\label{sec:mps}
This section is devoted to an overview of the algebraic modeling of the syndrome decoding problem proposed in~\cite{2021/meneghetti} (referred to as the MPS modeling). We fix the following notation for this section.

\begin{notation}\label{MPSnotation}
    Let $n\ge 2$ and let $\CC \subseteq \FF_2^n$ be a $[n,k,d]$-linear code having a parity check matrix $\HH \in \FF_2^{(n-k) \times n}$. We define $\ell = \lfloor \log_2(n) \rfloor + 1$. Let $\vs \in \FF_2^{n-k}$ play the role of the syndrome and let $0\le t \le \lfloor (d-1)/2 \rfloor$ be the target error weight. 
    Let $X = \left(x_1,\ldots,x_n\right)$ and $Y=(Y_1,\dots,Y_n)$ with $Y_j=(y_{j,1}, \dots, y_{j,\ell})$ be two sets of variables and we consider the polynomial ring $\FF_2[X,Y]$. 
    
\end{notation}

We define the following maps $\pi_i$ for $i=1,\ldots,n$,

\begin{align*}
    \pi_i : \FFt^{n} &\rightarrow \FFt^i \\
    (v_1,\ldots,v_n) &\mapsto (v_1,\ldots,v_i).
\end{align*}

The construction of the proposed algebraic modeling consists of four steps and uses the variables contained in $X$ and $Y$ to express relations and dependencies. Each of these steps produces a set of polynomials in $\FF_2[X,Y]$. An extra step of the construction reduces the aforementioned polynomials to quadratic polynomials.

The idea is to construct an algebraic system having a variety containing elements $(\vx \mid \vy_1 \mid \cdots \mid \vy_n)\in \FFt^{n(\ell + 1)}$ whose first $n$ entries represent an element $\vx$ of $\FFt^n$ such that $\HH\vx^\top = \vs^\top$. The remaining $n\ell$ entries are considered to be the concatenation of $n$ elements $\vy_i \in \FFt^{\ell}$ where the elements of $\vy_i$ represent the binary expansion of $\wt(\pi_i(\vx))$ for every $i=1,\ldots,n$, with $\pi_i(\vx)=(x_1,\dots,x_i)$. By this definition, the list $\vy_n$ represents the binary expansion of $\wt(\vx)$. The system finally enforces that $\vy_n$ represents the binary expansion of an integer $t^\prime$ such that $t^\prime \le t$.
The elements of the variety of solutions of this algebraic modeling are finally projected onto their first $n$ coordinates, revealing the solutions to the original syndrome decoding problem.

Here is a description of the four steps of reduction of the MPS modeling. We describe the set obtained in each step as a set of polynomials in $\FFt[X,Y]$.

\begin{itemize}
    \item \textit{Parity check encoding.} This step ensures that the solution of the algebraic system satisfies the parity check equations imposed by the parity check matrix $\HH$ and the syndrome vector $\vs$. Here, we compute the set of $n-k$ linear polynomials
    \begin{equation}\label{eq:pce}
        \left\{\sum_{i=1}^n h_{i,j}x_i + s_j \mid j\in\{1,\ldots,n-k\}\right\}.
    \end{equation}
    \item \textit{Hamming weight computation encoding.} This part of the modeling provides a set of polynomials that describes the binary encoding of $\wt(\pi_i(\vx))$ for every $i=1,\ldots,n$ described above. The set of polynomials achieving this goal, is given by the union of the three following sets consisting of the $\ell+n-1$ polynomials in the sets
    \begin{equation}
	\begin{split}\label{eq:lineareqs}
        &\left\{ f_{1,1}=x_1 + y_{1,1}, f_{1,2}=y_{1,2}, \ldots, f_{1,\ell}=y_{1,\ell} \right\},\\ &\left\{f_{i,1}=x_i + y_{i, 1} + y_{i-1,1} \mid i=2,\ldots,n \right\}
    \end{split}
    \end{equation}
	and the $(n-1)(\ell -1)$ polynomials
    \begin{equation}\label{eq:othereqs}
            \left\{ f_{i,j}=\left(\prod_{h=1}^{j-1}y_{i-1, h}\right)x_i + y_{i,j} + y_{i-1,j} \mid i=2,\ldots,n,\ j=2,\ldots,\ell \right\}.
    \end{equation}
    We labeled the polynomials of the sets in~\eqref{eq:lineareqs} and in~\eqref{eq:othereqs} because the improvements in the next sections will mainly involve them.
    \item \textit{Weight constraint encoding.} This part produces a set consisting of a single polynomial that enforces the constraint $\wt(\vx) \le t$ by dealing with the variables in $Y_n$. Let $\vv \in \FFt^\ell$ represent the binary expansion of $t$. Consider the $\ell$ polynomials in $\FFt[X,Y]$ defined as
    $$f_j = (y_{n, j} +v_j)\prod_{h=j+1}^\ell (y_{n, h} + v_h + 1) $$
    for $j=1,\ldots,\ell$.
    The set is the singleton 
    \begin{equation}\label{eq:MPSwce}
        \left\{ \sum_{j=1}^\ell (v_j + 1)f_j \right\}.  
    \end{equation}
    
    \item \textit{Finite field equations.} The set of $n + n\ell$ finite field polynomials of $\FFt[X,Y]$ is
    \begin{equation} \label{eq:ffe}
        \left\{x_i^2- x_i \mid i=1,\ldots,n\right\} \cup \left\{y_{i,j}^2- y_{i,j} \mid i=1,\ldots,n,\ j=1,\ldots,\ell\right\},
    \end{equation}
    and ensures that the elements of the variety are restricted to elements of $\FFt^{n(\ell + 1)}$.
\end{itemize}

The algebraic system corresponding to an instance of the syndrome decoding problem is then the union of the four sets described above.
Clearly, this is not a quadratic system; thus the authors apply a linearization strategy that introduces a number of auxiliary variables used to label monomials of degree $2$. This eventually results in a large quadratic system in many more than just $n(\ell + 1)$ variables. In fact, the final quadratic system ends up having equations and variables bounded by $\OO(n\log_2(n)^2)$.

\section{Improving the MPS Modeling}\label{sec:EWM}
In this section, we provide improvements of the MPS modeling that reduce the number of equations and variables in the final algebraic system. We keep the same notation as in Notation~\ref{MPSnotation}.
First, we consider the case of the syndrome decoding problem, i.e. with a bounded weight error. We then consider the case of the exact weight syndrome decoding problem. We observe that one can avoid the linearization step as the resulting system is already quadratic.

\subsection{Improved Modeling for the Case of SDP}\label{subsec:f2SD}
We consider the $\mathsf{degrevlex}$ monomial ordering on $\FFt[X,Y]$ with the $X$ variables greater than the $Y$ variables, and denote by $\lm(p)$ the leading monomial  of a polynomial $p$. Notice that since we are in the binary case, the notions of leading monomial and that of leading term coincide.

Denote by $F = \{f_{i,j} \mid i=1,\ldots,n,\ j=1,\ldots,\ell\} \subset \FFt[X,Y]$ the set of polynomials of cardinality $n\ell$ given by \eqref{eq:lineareqs} and \eqref{eq:othereqs} %\mathsf{hwce}
for a code of length $n$. 
We aim at building a set $G=\{g_{i,j} \mid i=1,\ldots,n,\ j=1,\ldots,\ell\}\subset \FFt[X,Y]$ consisting of polynomials of degree at most $2$  such that $\langle G \rangle = \langle F \rangle$. Denote with $F[i,j]$ the polynomial $f_{i,j}$, similarly for $G$. We first give a description of the set $G$ and then formally describe the new modeling.    

Construct $G$ as follows:
\begin{itemize}
    \item Put $G[1,1] = x_1 + y_{1,1}$ and $G[1,h] = y_{1,h}$ for $h = 2,\ldots, \ell$;
    \item Set $G[i,1] = F[i,1] = x_i + y_{i, 1} + y_{i-1,1}$ for every $i = 2,\ldots,n$;
    \item Compute
\begin{align*}
    G[i,j] &= F[i,j] + y_{i-1, j-1}F[i,j-1]\\
    &= F[i,j] + \lm(F[i,j]) + y_{i-1, j-1}(y_{i,j-1} + y_{i-1,j-1})\\
    &= y_{i,j} + y_{i-1,j} + y_{i-1,j-1}^2 + y_{i,j-1}y_{i-1,j-1}.
    %&= x_{n+(i-1)\ell+j-1}x_{n+(i-2)\ell+j-1} + x_{n+(i-1)\ell+j} + x_{n+(i-2)\ell+j} + x_{n+(i-1)\ell+j-1},
\end{align*}
    for every $i=2,\ldots,n$ and $j = 2,\ldots,\ell$, where equality holds because $\lm(F[i,j]) = y_{i-1,j-1}\lm(F[i,j-1])$.
\end{itemize}

\begin{remark}
    The algebraic system we are going to construct contains the field polynomials $x_i^2- x_i$ for each $i=1,\ldots,n$ and $y_{i,j}^2- y_{i,j}$ for every $i=1,\ldots,n$ and  $j=1,\ldots,\ell$. Therefore, in terms of generating elements of the ideal, any squared term in $G[i,j]$  can be reduced to a linear term.
\end{remark}

The set $G \subset \FFt[X,Y] $ contains $n\ell$ polynomials of degree at most two.
The following proposition proves that the set $G \subset \FFt[X,Y]$ computed as above and $F$ generate the same ideal of $\FFt[X,Y]$.
\begin{proposition}
    We have $\langle G \rangle = \langle F \rangle$.
\end{proposition}
\begin{proof}
    The inclusion $\langle G \rangle \subseteq\langle F \rangle$ is trivial.    
    To prove the other inclusion, we show that we can write any element of the basis $F$ as an $\FFt[X,Y]$-linear combination of elements of the basis $G$.
    By construction, $G[1,j] = F[1,j]$ for every $j=1,\ldots,\ell$. For every $i = 2,\ldots,n$ we prove $F[i,j]\in \langle G \rangle$ by induction on $j$.\\
    For $j=1$ we have $F[i,1] = G[i,1]$.\\
    Assume that $F[i,j] = \sum_{h=1}^j p_{i,j,h} G[i,h]$ with $p_{i,j,h}\in \FFt[X,Y]$. Then by construction we have
    \begin{align*}
        F[i,j+1] &= G[i,j+1] - y_{i-1, j}F[i,j]\\
        &= G[i,j+1] - y_{i-1, j} \sum_{h=1}^j p_{i,j,h} G[i,h]
    \end{align*}
    proving the claim.
    \qed
\end{proof}
We thus redefine the Hamming weight computation encoding as follows:
\begin{itemize}
    \item \textit{Hamming weight computation encoding.} Compute the following union of subsets of $\FFt[X,Y]$:
    \begin{align*}
         &\left\{ x_1 + y_{1,1}, y_{1,2}, \ldots, y_{1,\ell} \right\} \cup \left\{x_i + y_{i, 1} + y_{i-1,1} \mid i=2,\ldots,n \right\}\\
 &\cup  \big\{ y_{i,j-1}y_{i-1,j-1} + y_{i,j} + y_{i-1,j-1} + y_{i-1,j} \\ & \ \ \ \mid i=2,\ldots,n,\ j=2,\ldots,\ell \big\},
    \end{align*}
\end{itemize}

\subsubsection{Further improvement.}
Set now $\ell_t = \lfloor \log_2 (t) \rfloor + 1$. A further improvement to the MPS modeling (described in Equation~\eqref{eq:SDhwce}) follows by observing that in the non-trivial case where $t < n$, we can impose that the last $\ell-\ell_t$ entries of $\vy_i$  must be $0$ for every $i=1,\ldots,n$. This means that we can add the linear equations $y_{i, j} = 0$ for every $i=1,\ldots,n$ and $j=\ell_t+1,\ldots,\ell$. 
By inspection, setting the aforementioned variables to $0$ will make part of the equations of the Hamming weight computation encoding vanish. We can equivalently simply consider the equations that remain, and get rid of the variables which have been set to $0$. 
Consider the following updated notation.

\begin{notation}\label{ImprovedMPSnotation}
    Let $n\ge 2$ and let $\CC \subseteq \FF_2^n$ be a $[n,k,d]$-linear code having a parity check matrix $\HH \in \FF_2^{(n-k) \times n}$.  Let $\vs \in \FF_2^{n-k}$ play the role of the syndrome and let $0\le t \le \lfloor (d-1)/2 \rfloor$ be the target error weight. We define $\ell_t = \lfloor \log_2(t) \rfloor + 1$. Let $X = \left(x_1,\ldots,x_n\right)$ and $Y=(Y_1,\dots,Y_n)$ with $Y_j=(y_{j,1}, \dots, y_{j,\ell_t})$ be two sets of variables and consider the polynomial ring $\FF_2[X,Y]$. 
   % Let $X = \{x_1,\ldots,x_n\}$ and $Y=Y_1\cup\dots\cup Y_n$ with $Y_j=\{y_{j,1}, \dots, y_{j,\ell_t}\}$ be two sets of variables and consider the polynomial ring $\FF_2[X,Y]$.
\end{notation}

Under Notation~\ref{ImprovedMPSnotation}, the effect of our improvement on the set of polynomials produced by the Hamming weight computation encoding is the following.
\begin{itemize}
    \item \textit{Hamming weight computation encoding.} Compute the following union of subsets of $\FFt[X,Y]$:
    \begin{equation}\label{eq:SDhwce}
    \begin{split}
         &\left\{ x_1 + y_{1,1}, y_{1,2}, \ldots, y_{1,\ell_t} \right\} \cup \left\{x_i + y_{i, 1} + y_{i-1,1} \mid i=2,\ldots,n \right\}\\
 &\cup  \big\{ y_{i,j-1}y_{i-1,j-1} + y_{i,j} + y_{i-1,j-1} + y_{i-1,j}  \\ & \ \ \ \mid i=2,\ldots,n,\ j=2,\ldots,\ell_t \big\} \cup \left\{ y_{i,\ell_t}y_{i-1,\ell_t} + y_{i-1,\ell_t} \mid i=2,\ldots,n\right\}.
    \end{split}
    \end{equation}
\end{itemize}

The effect on the weight constraint encoding is simply the decrease in the degree from $\ell$ to $\ell_t$ of the produced polynomial. This is the only non-quadratic polynomial left in the modeling. We can turn this polynomial into a set of $\OO(t\ell_t)$ polynomials of degree up to $2$ in $\OO(t\ell_t)$ variables with the same linearization techniques described in~\cite[Fact 1 and Lemma 11]{2021/meneghetti}.

To summarize, our modeling is defined in the following way.
\begin{modeling}[Improved Modeling for the SDP over $\FF_2$] \label{modeling: improvedSD_F2}
   Given an instance $(\HH,\mathbf{s},t)$ of Problem~\ref{BSDP} over $\FF_2$, Modeling~\ref{modeling: improvedSD_F2} is the union of the sets of polynomials \eqref{eq:pce},\eqref{eq:MPSwce}, \eqref{eq:ffe} and \eqref{eq:SDhwce}.
\end{modeling}

The improved modeling is an algebraic system of $\OO(n(\ell_t+2) -k + t\ell_t)$ polynomials of degree at most $2$ in $\OO(n(\ell_t+1) + t\ell_t)$ variables. Note that most applications of the SDP to code-based cryptography, for instance in the McEliece scheme, choose $t \ll n$, hence the asymptotic bounds on the number of polynomials and variables in the improved modeling are both $\OO(n\ell_t)$.
As shown in Table \ref{table: improvement}, our modeling improves over MPS by a factor of $\log_2(n) \log_t(n)$.
\begin{table}[H]
    \centering
\begin{tabular}{|c|c|c|}
	\hline
	& \# Polynomials & \# Variables\\
	\hline 
	\cite{2021/meneghetti}		& $\mathcal{O}( n \log_2(n)^2)$ & $\mathcal{O}( n \log_2(n)^2)$ \\
 \hline
	Modeling~\ref{modeling: improvedSD_F2} 	& $\OO(n\log_2(t))$ & $\OO(n\log_2(t))$\\
			\hline
	\end{tabular}
  \vspace{2mm}
    \caption{Comparison with the asymptotic size of the polynomial system in \cite[Theorem 13]{2021/meneghetti}, where $n$ is the length of the code and $t$ the bound on the weight of the target vector, that is  $\wt(\ve)\leq t$.}
    \label{table: improvement}
\end{table}

\subsection{Improved Modeling for the Case of ESDP}\label{subsec:f2ESD}
It is possible to obtain an algebraic modeling for the ESDP by tweaking the modeling described in the previous section. 
In fact, it is enough to redefine the weight constraint encoding to enforce that $\vy_n$ represents the binary expansion of an integer $t^\prime$ such that $t^\prime=t$ exactly. To this end, let $\vv \in \FFt^{\ell_t}$ represent the binary expansion of an integer $t$.  Under the same notation as in Notation~\ref{ImprovedMPSnotation}, the following version of the weight constraint encoding describes the ESDP modeling with $\wt(\ve) = t$.
\begin{itemize}
    \item \textit{Weight constraint encoding.} Compute the following set of linear polynomials: 
    \begin{equation}\label{eq:ESDwce}
        \left\{ y_{n, j} + v_j \mid j=1,\ldots,\ell_t \right\}.
    \end{equation}
\end{itemize}

Using these polynomials leads to Modeling
\begin{modeling}[Improved Modeling for the ESDP over $\FF_2$] \label{modeling: improvedESD_F2}
   Given an instance $(\HH,\mathbf{s},t)$ of Problem~\ref{EWSDP} over $\FF_2$, Modeling~\ref{modeling: improvedESD_F2} is the union of the sets of polynomials \eqref{eq:pce}, \eqref{eq:ffe}, \eqref{eq:SDhwce} and \eqref{eq:ESDwce}.
\end{modeling}

Observe that, replacing the original Hamming weight computation encoding with that in~\eqref{eq:SDhwce} and the weight constraint encoding with that in~\eqref{eq:ESDwce}, we obtain an algebraic system of polynomials of degree at most $2$ for ESDP. Hence, linearization is not needed, moreover, we can give the exact number of equations and variables of this system. We report these values in Table~\ref{table:esd-model-sizes}.

\begin{table}[H]
    \centering
\begin{tabular}{|c|c|c|}
	\hline
	& \# Polynomials & \# Variables\\
 \hline
	
   Modeling~\ref{modeling: improvedESD_F2}  & $2n\ell_t + 3n + \ell_t - k - 1$ & $n(\ell_t + 1)$\\
   \hline
	\end{tabular}
  \vspace{2mm}

    \caption{Number of equations and variables of the algebraic modeling of ESDP with $\wt(\ve)=t$. The value of $\ell_t$ is $\lfloor \log_2(t) \rfloor + 1$.}
    \label{table:esd-model-sizes}
\end{table}

%%%%%%%%%%%%%%%%%%%%%%%%%%%%%%%
%%%%%%%%%%%%%%%%%%%%%%%%%%%%%%%
%\section{Complexity Analysis}
\section{Complexity Analysis of Modeling~\ref{modeling: improvedESD_F2}}\label{sec:complexity-analysis}  

\label{sec:dreg-EWM}
In this section, we investigate the complexity of solving the algebraic system for the ESDP given in Modeling~\ref{modeling: improvedESD_F2} using standard Gröbner basis methods. An upper bound on the complexity is given by the formula \eqref{eq:GBcomplexity} which depends on both the number of variables and the solving degree.  Typically, the solving degree of the system is estimated by assessing its degree of regularity.

However, in our analysis, we experimentally show that the degree of regularity often significantly exceeds the solving degree for systems given in Section~\ref{subsec:f2ESD} (see the results in Table~\ref{Tab:q2-SolveDeg}).
This distinction is crucial in cryptography, where these concepts are frequently used interchangeably. Our findings underscore the importance of thoroughly verifying such claims to ensure accurate security assessments and parameter selection.

\begin{remark}
We point out that the study in \cite{2023/briaud} investigates a particular case of the problem that this paper deals with, that is the \emph{regular} syndrome decoding problem. The regular syndrome decoding problem considers error vectors having a regular distribution of non-zero entries. The algebraic modeling proposed in~\cite{2023/briaud} is conjectured to exhibit semi-regular behavior when the linear parity-check constraints and the fixed, structured quadratic polynomials are considered separately. This suggests that, to some extent, their model behaves like a random polynomial system. Despite the fact that the problem tackled in~\cite{2023/briaud} is a particular case of the problem we consider, our modeling has not been devised as a generalization of their modeling. Furthermore, we show that for the more general case, our modeling yields different results.
\end{remark}

	For the rest of this section, we retain the notation defined in Notation~\ref{ImprovedMPSnotation}. 
    We consider the polynomial ring  $\FFt[X,Y]$ with the $\mathsf{degrevlex}$ term order with the $X$ variables greater than the $Y$ variables. Let $S \subset \FFt[X,Y]$ be the set of polynomials of Modeling~\ref{modeling: improvedESD_F2} as described in Section~\ref{subsec:f2ESD}. 
			Let $L$ and $Q$ denote the sets of linear and quadratic polynomials, respectively. Clearly $S = L \cup Q$. Write also $L = L_\vH \cup P$, where $L_\vH$ denotes the set of linear polynomials in~\eqref{eq:pce} introduced with the parity check matrix $\vH$, and $P$ denotes the remaining linear polynomials in $S$. In other words, $P$ is the following set
			\[\begin{split}
				P = &\left\{ x_1 + y_{1,1}, y_{1,2}, \ldots, y_{1,\ell_t} \right\} \cup \left\{x_i + y_{i, 1} + y_{i-1,1} \mid i=2,\ldots,n \right\} \\ \cup &\left\{ y_{n, j} + v_j \mid j=1,\ldots,\ell_t \right\}.
                \end{split}
			\]
			We want to estimate the degree of regularity of $S$. 
			Since we do not know $L_\vH$ a priori, we consider the set $S\setminus L_\vH = Q \cup P$ and compute its degree of regularity. Indeed, we found that analyzing the degree of regularity or solving degree of the system with the linear equations \eqref{eq:pce} of $L_\vH$ included was too challenging and unpredictable, as it heavily depends on the specific instance of the parity check matrix $\vH$.
For this reason, we chose to establish mathematical results for the system without $L_{\vH}$, with the aim of providing a clearer foundation. Notice that the degree of regularity of $S\setminus L_\vH = Q \cup P$ gives an upper bound to the degree of regularity of the whole system $S$ (see Remark~\ref{rem:range fordregS}).

			We break down the problem by first computing the degree of regularity of $Q$ and then that of $Q \cup P$.
			We take advantage of the fact that the Hilbert series of $Q$ and of $Q \cup P$ are polynomials and compute their degree, i.e. for instance, $\dreg{Q}=\deg \mathrm{HS}_{\FFt[X,Y]/\langle Q^\htop\rangle}(z)+1$ as per Remark~\ref{rem:polyHS}, similarly for $Q\cup P$. To this end, we are going to compute the maximum degree of a monomial in $\FFt[X,Y]/\langle Q^\htop\rangle$, similarly we do for $Q \cup P$.

\subsubsection{The quadratic polynomials.}\label{subsec:quad-polys}
We begin by studying the degree of regularity of the quadratic part $Q$ of the system $S$ of  Modeling~\ref{modeling: improvedESD_F2}. The highest degree part of $Q$ has a very nice structure, as explained in the following remark.

\begin{remark}\label{rem:qtopdef}
    The set $Q^\htop$ is the union of the following three sets
    $$\left\{x_i^2 \mid i=1,\ldots,n\right\}, \left\{y_{i,j}^2 \mid i=1,\ldots,n,\ j=1,\ldots,\ell_t\right\}$$ 
    and $$\left\{ y_{i-1,j}y_{i,j}  \mid i=2,\ldots,n,\ j=1,\ldots,\ell_t \right\}.$$ The ideal $\langle Q^\htop \rangle \subseteq \FFt[X,Y]$ is thus a monomial ideal.
\end{remark}

The following lemma gives the structure of the quotient ring $\FFt[X,Y]/\langle Q^\htop \rangle$.
\begin{lemma}\label{lem:groebnerQh}
    The set $Q^\htop$ is a Gr\"obner basis of the ideal $\langle Q^\htop\rangle$.
\end{lemma}
\begin{proof}
    As observed in Remark~\ref{rem:qtopdef}, $Q^\htop$ is a monomial ideal. Given any two elements of $m_1,m_2 \in Q^\htop$ it is clear that for $a = \lcm (m_1,m_2)/m_1 \in \FFt[X,Y]$ and $b = \lcm (m_1,m_2)/m_2 \in \FFt[X,Y]$ we have that $am_1 - bm_2 = 0$. \qed
\end{proof}
\ifodd0
    We can exploit the knowledge of the Gr\"obner basis of $\langle Q^\htop \rangle$ given in Lemma \ref{lem:groebnerQh} to compute the coefficients of the Hilbert series $\mathcal{H}_R$. The $(k+1)$-th coefficient of $\mathcal{H}_R$ is given by $\dim_{\FFq}(\FFt[X,Y]_k/I_k)$, in other words, the number of monomials of degree $k$ in $R$. This coincides with the number of monomials of $\FFt[X,Y]$ of degree $k$ that are not a multiple of any monomial in $\GG$.

    We can model this problem in terms of subsets of $[n(l+1)]$, or equivalently, elements of $2^{[n(l+1)]}$. Let $B_1,\ldots
    B_{n\ell -n-\ell +1}$ be the sets of two elements indexing the variables of each mixed monomial in $\GG$ (monomials in the third set). Counting monomials of degree $k$ in $R$ boils down to counting the number of subsets of $[n(l+1)]$ of cardinality $k$ not containing any $B_i$.
\fi
       \begin{example}\label{ex:n4}
        Let $n=4$ be the length of a code, then $\ell_t = 2$. A Gr\"obner basis of $\langle Q^\htop \rangle$ is the union of
       \begin{equation*}
            \left\{ y_{1,1}y_{2,1},
            y_{1,2}y_{2,2},
            y_{2,1}y_{3,1},
            y_{2,2}y_{3,2},
            y_{3,1}y_{4,1},
            y_{3,2}y_{4,2}\right\}
       \end{equation*}
       and
       \begin{equation*}
           \left\{
            x_{1}^2,
            x_{2}^2,
            x_{3}^2,
            x_{4}^2,
            y_{1,1}^2,
            y_{1,2}^2,
            y_{2,1}^2,
            y_{2,2}^2,
            y_{3,1}^2,
            y_{3,2}^2,
            y_{4,1}^2,
            y_{4,2}^2
            \right\}.
       \end{equation*}
       \ifodd0
       Following our argument we obtain the $(n-1)\cdot(l-1) = n\ell -n-\ell+1 = 6$ sets $B_i$, indexing mixed monomials, are
       \begin{align*}
           B_1 = \{1,4\},&B_2 = \{4,7\},B_3 = \{7,11\},\\
           B_4 = \{2,5\},&B_5 = \{5,8\},B_6 = \{8,11\}.
       \end{align*}
       \fi
   \end{example}

\noindent The following simple lemma is crucial for computing the degree of regularity of $Q$. For the sake of simplicity, we state it in terms of sets, and it ultimately provides a method to construct maximal monomials in the quotient ring $\FFt[X,Y]/\langle Q^\htop \rangle$.

\begin{lemma}\label{lem:maximalset}
Let $ \mathcal{N} = \{1, 2, 3, \dots, n\} $ and $ \mathcal{P} = \{\{1,2\}, \{2,3\}, \dots, \{n-1, n\}\} $, where $ \mathcal{P} $ consists of consecutive pairs of elements from $ \mathcal{N} $. Then:
\begin{itemize}
    \item If $ n $ is even, there are exactly two sets of maximal cardinality  $ \mathcal{S}_1, \mathcal{S}_2 \subseteq \mathcal{N} $ such that no set in $ \mathcal{P} $ is a subset of $ \mathcal{S} $.
    \item If $ n $ is odd, there is exactly one set of maximal cardinality $ \mathcal{S} \subseteq \mathcal{N} $ such that no set in $ \mathcal{P} $ is a subset of $ \mathcal{S} $.
\end{itemize}
\end{lemma}
\begin{proof}
We aim to find the number of sets of maximal cardinality  $ \mathcal{S} \subseteq \mathcal{N} $ such that no pair from $ \mathcal{P} $ (i.e., no two consecutive elements) appears in $ \mathcal{S} $.
In order to avoid pairs of consecutive elements, we can only select non-consecutive elements from $ \mathcal{N} $. To maximize the size of $ \mathcal{S} $, we select every other element from $ \mathcal{N} $. The size of such a set of maximal cardinality $ \mathcal{S} $ is: $\left\lceil \frac{n}{2} \right\rceil$. Thus:
\begin{itemize}
    \item If $ n $ is even, a set of maximal cardinality contains $ \frac{n}{2} $ elements.
    \item If $ n $ is odd, a set of maximal cardinality contains $ \frac{n+1}{2} $ elements.
\end{itemize}
\textbf{Case 1: $ n $ is even.} 
Let $ n = 2k $. The largest possible set $ \mathcal{S} $ will contain $ k = \frac{n}{2} $ elements. 
There are exactly two ways to construct such a set:
\begin{enumerate}
    \item Start with 1 and select every other element: $\mathcal{S}_1 = \{1, 3, 5, \dots, n-1\}.$
    This set contains all the odd-numbered elements of $ \mathcal{N} $, and its size is $ k $.
    
    \item Start with 2 and select every other element: $\mathcal{S}_2 = \{2, 4, 6, \dots, n\}.$
    This set contains all the even-numbered elements of $ \mathcal{N} $, and its size is also $ k $.
\end{enumerate}
Since there are no other ways to select $ k $ elements without picking consecutive elements, these are the only two sets of maximal cardinality  for $ n $ even.\\
\textbf{Case 2: $ n $ is odd.}
Let $ n = 2k + 1 $. The largest possible set $ \mathcal{S} $ contains $ k + 1 = \frac{n+1}{2} $ elements.
In this case, there is only one way to construct a set of size $ k + 1 $ that avoids consecutive elements, i.e. start with 1 and select every other element: $\mathcal{S}_1 = \{1, 3, 5, \dots, n\}.$
This set contains $ k + 1 $ elements and avoids consecutive pairs.
If we were to start with 2 and select every other element, we would only get $ k $ elements: $\mathcal{S}_2 = \{2, 4, 6, \dots, n-1\}.$
This is not maximal, as it contains fewer than $ k + 1 $ elements.
Thus, for $ n $ odd, there is exactly one maximal set.
\qed
\end{proof}
Lemma~\ref{lem:maximalset} can be used to prove the following corollary, which we will use to construct a maximal degree monomial in $\FFt[X,Y]/\langle Q^\htop \rangle$.
The idea behind the construction lies in the observation that a Gr\"obner basis of $Q^\htop$ can be written as the union of disjoint subsets $Q^\htop_{j,n}$ for $j=1,\ldots,\ell_t$, see Theorem~\ref{Thm:Dreg-of-Qtop}, which we describe in the next corollary. Also, the next corollary computes a maximal degree monomial with respect to $Q^\htop_{j,n}$ for every $j=1,\ldots,\ell_t$. Given these monomials, computing a maximal degree monomial in $\FFt[X,Y]/\langle Q^\htop \cup P^\htop\rangle$, or equivalently, the degree of its Hilbert series, becomes feasible with a slight modification of the subsets due to the presence of linear polynomials in $P^\htop$.

\begin{corollary}\label{cor:maximalmonomial}
    Let $n\in \mathbb{N}$ with $n\ge 2$, and define $$Q^\htop_{j,n} := \left\{ y_{1,j}y_{2,j}, y_{2,j}y_{3,j}, \ldots, y_{n-1,j}y_{n,j}\right\} \cup \left\{y_{i,j}^2 \mid i=1,\ldots,n\right\} \subset \FFt[y_{1,j},\ldots,y_{n,j}],$$
    for some $j\in \mathbb{N}$.
    If $n$ is even then there exists two monomials of maximal degree $\left\lceil\frac{n}{2} \right\rceil$ in $\FFt[y_{1,j},\ldots,y_{n,j}]/\langle Q^\htop_{j,n} \rangle$, namely \[
        m_1 = \prod_{\substack{i=1,\ldots,n-1,\\ i\ \text{odd}}}y_{i,j} \quad \textnormal{and}\quad m_2 =\prod_{\substack{i=2,\ldots,n,\\ i\ \text{even}}}y_{i,j}.
    \]
    If $n$ is odd, then there exists a unique monomial of maximal degree $\left\lceil\frac{n}{2} \right\rceil$ in $\FFt[y_{1,j},\ldots,y_{n,j}]/\langle Q^\htop_{j,n} \rangle$, namely
    \[
        m = \prod_{\substack{i=1,\ldots,n,\\ i\ \text{odd}}}y_{i,j}.
    \]
\end{corollary}

\noindent We are ready to prove the following theorem, which provides the degree of regularity of $Q$. 

\begin{theorem}\label{Thm:Dreg-of-Qtop}
      $$\dreg{Q}= \begin{cases}
        n + \ell_t n/2 + 1 \quad &\text{ if } n \equiv 0 \bmod 2\\
        n + \ell_t(n+1)/2 + 1 \quad &\text{ if } n \equiv 1 \bmod 2
    \end{cases}.$$
    Equivalently,
    $$\dreg{Q} = n + \ell_t\lceil n/2 \rceil  + 1.$$
\end{theorem}

\begin{proof}
    Let $Q^\htop_{j,n} \subset \FFt[y_{1,j},\ldots,y_{n,j}]$ as in Corollary~\ref{cor:maximalmonomial}, for every $j=1,\ldots,\ell_t$. Observe that 
    \begin{equation}\label{eq:qtopasunion}
    Q^\htop = \bigcup_{j=1}^{\ell_t} Q^\htop_{j,n} \cup \left\{x_i^2 \mid i=1,\ldots,n\right\}.
    \end{equation}
    Corollary~\ref{cor:maximalmonomial} computes a monomial $m_j \in \FFt[y_{1,j},\ldots,y_{n,j}]$ of maximal degree $\lceil n/2 \rceil$ such that $m_j \not \in \langle Q^\htop_h\rangle$ for every $j=1,\ldots,\ell_t$ and every $h=1,\ldots,\ell_t$. This implies that $m_j \not \in \langle Q^\htop \rangle$ for every $j$.
    It is now clear that the monomial
    \[
        m:= \prod_{i=1}^n x_i \prod_{j=1}^{\ell_t}m_j \in \FFt[X,Y]
    \]
    is such that $m \not \in \langle Q^\htop \rangle$. Note that the the set $\left\{x_i^2 \mid i=1,\ldots,n\right\}$ in \eqref{eq:qtopasunion} enforces that $m$ must be squarefree in the variables $x_1,\ldots,x_n$. By the maximality of each $m_j$ and that of $\prod_{i=1}^n x_i$, any multiple of $m$ by a non-constant term would trivially be in $\langle Q^\htop \rangle$.
    Since $$d:=\deg m = n + \ell_t\lceil n/2 \rceil,$$
    we have that the $(d+1)$-th coefficient of the Hilbert series of  $\FFt[X,Y]/\langle Q^\htop \rangle$ is $0$. The result on the degree of regularity $\dreg{Q}$ follows.
    \qed
\end{proof}

\ifodd0

\begin{theorem}\label{Thm:Dreg-of-Qtop}
      $$\dreg{Q}= \begin{cases}
        2n + (\ell-1)n/2 + 1 \quad &\text{ if } n \equiv 0 \bmod 2\\
        2n + (\ell-1)(n+1)/2 + 1 \quad &\text{ if } n \equiv 1 \bmod 2
    \end{cases}.$$
    Equivalently,
    $$\dreg{Q} = 2n + (\ell-1)\lceil n/2 \rceil  + 1.$$
\end{theorem}

\begin{proof}
We set $R=\FFt[X,Y]/\langle Q^\htop \rangle$.
    We show that the maximum degree of a monomial of $R$ is $d := 2n + (l-1)\lceil n/2 \rceil$ via a constructive argument and show that any monomial of degree $d+1$ lives in $\langle Q^\htop \rangle$.
    From Lemma~\ref{lem:groebnerQh} and Remark~\ref{rem:qtopdef} a Gr\"obner basis of $\langle Q^\htop \rangle$ is $Q^\htop$ itself, i.e. the union of three sets of quadratic monomials, the first two being the squares of all the variables $\vx$ and $\vy$ and the third being the set of mixed monomials $$\left\{ y_{(i-1)\ell+j-1}y_{(i-2)\ell+j-1}  \mid i=2,\ldots,n,\ j=2,\ldots,\ell\right\}.$$
    According to the first two sets, to construct a monomial of $R$ we need to build a square free monomial, i.e. which is at most linear in each variable. The third set imposes restrictions on the variables that can be chosen to build the monomial.

    Assume that $n \equiv 0 \bmod 2$, the case with $n$ odd can be proven similarly, by a different choice of variables.
    Consider the set of variables $$V_1 := \left\{y_{(h(\ell-1)-2)\ell +j - 1} \mid h=1,\ldots,n/2,\  j=2,\ldots,\ell\right\}$$ and let $M_1:= \prod_{v\in V_1}v$. Furthermore, none of the variables in the set $$V_2 := \left\{x_i,y_{i\ell} \mid i=1,\ldots,n\right\}$$ appears in a mixed monomial. Let $M_2 := \prod_{v\in V_2}v$. Then $M := M_1M_2$ is a square-free monomial in $R$ of degree $d$. It is easy to verify that $M$ is not a multiple of any of the mixed monomials in the Gr\"obner basis of $\langle Q^\htop \rangle$ and that $d$ is maximal. In addition, any monomial of maximal degree can be constructed similarly.

    Indeed, any multiple of $M$ is also the multiple of a monomial in the Gr\"obner basis of $\langle Q^\htop \rangle$. To see this, note that, on the one hand, if we take a multiple of $yM$ with $v| M$, then $M$ is no longer square-free and thus $v^2 | M$ with $v^2$ lying on the Gr\"obner basis of $\langle Q^\htop \rangle$.
    On the other hand, the variables that do not appear in $M$ are exactly those in the set $$V_3 := \left\{y_{(h(\ell-1)-1)\ell +j - 1} \mid h=1,\ldots,n/2,\  j=2,\ldots,\ell\right\}.$$
    Let $v \in V_3$, and consider the monomial $vM$. We can write $v = y_{(h(\ell-1)-1)\ell +j - 1}$ for some $h\in [1,n/2]$ and $j=2,\ldots,\ell$. By construction of $M$, we have that $vw | vM$ where $w = y_{(h(\ell-1)-2)\ell +j - 1}$. Setting $i = h(\ell -1)$ we find that the product $vw = y_{(i-1)\ell+j-1}y_{(i-2)\ell+j-1}$ is a mixed monomial on the Gr"obner basis of $\langle Q^\htop \rangle$ meaning that $vM \equiv 0$ in $R$.
    Finally, notice that any divisor of degree $d' < d$ of $M$ is a nonzero monomial of $R$.
    This proves that the $d$-th coefficient of the Hilbert series of $R$ is positive and that, by maximality of $d$, the $(d+1)$-th coefficient is the first $0$ coefficient of the series, proving our claim.
    \qed
\end{proof}
\fi

\begin{example}
     Let $n=8$ and $\ell_t=3$. 
    According to Theorem~\ref{Thm:Dreg-of-Qtop} the degree of regularity of $ Q$ is $21$.
   % $$\dreg{Q}=8 + 3\left\lceil\frac{8}{2} \right\rceil + 1 = 21.$$
    \noindent Using MAGMA, we compute and report the Hilbert series of the quotient  ring $\FFt[X,Y]/\langle Q^\htop\rangle$, i.e.
    \begin{align*}
        \mathrm{HS}_{\FFt[X,Y]/\langle Q^\htop\rangle} (z) =&\ 125z^{20} + 2500z^{19} + 23075z^{18} + 130800z^{17} + \\
        &\ 511140z^{16} + 1465020z^{15} + 3198081z^{14} +\\
        &\ 5448312z^{13} + 7360635z^{12} + 7966528z^{11} + \\
        &\ 6946904z^{10} + 4889800z^9 + 2773415z^8 + \\
        &\ 1260580z^7 + 454625z^6 + 128080z^5 + 27524z^4 + \\
        &\ 4348z^3 + 475z^2 + 32z + 1,
    \end{align*}
   thus $\dreg{Q}=\deg \mathrm{HS}_{\FFt[X,Y]/\langle Q^\htop\rangle}+1=21$, matching our results.
\end{example}

%%%%%%%%%%%%%%%%%%%%%%%%%%%%%%%    
\subsubsection{The linear polynomials.}
In this section, we study how the degree of regularity computed in Theorem~\ref{Thm:Dreg-of-Qtop} changes when we add to the quadratic equations $Q$ also the fixed linear equations of $P$, which do not depend on the specific instance of the problem. Specifically, we compute the degree of regularity of $Q \cup P$. For this, we need to consider the ideal $\langle Q^\htop \cup P^\htop\rangle$. Note that this ideal contains $\langle Q^\htop \rangle$, which means that the variety of the former is a subset of the variety of the latter. In particular, the ideal $\langle Q^\htop \cup P^\htop\rangle$  is also zero-dimensional, so its degree of regularity is well-defined. We will use similar arguments to those applied to $\langle Q^\htop \rangle$ to study it. 

\begin{remark}\label{rem:qtopptopdef}
    The set $Q^\htop \cup P^\htop$ is the union of the following sets
    $$\left\{x_i^2 \mid i=1,\ldots,n\right\},\left\{x_i \mid i=1,\ldots,n\right\}, \left\{y_{i,j}^2 \mid i=1,\ldots,n,\ j=1,\ldots,\ell_t\right\},$$ 
    $$\left\{y_{1,j} \mid j=2,\ldots,\ell_t\right\},\left\{y_{n,j} \mid j=1,\ldots,\ell_t\right\}$$
    and
    $$\left\{ y_{i-1,j}y_{i,j}  \mid i=2,\ldots,n,\ j=1,\ldots,\ell_t \right\}.$$ and the ideal $\langle Q^\htop \cup P^\htop \rangle \subseteq \FFt[X,Y]$ is thus a monomial ideal.
\end{remark}
Next lemma provides a Gr\"obner basis of the ideal $\langle Q^\htop \cup P^\htop \rangle \subseteq \FFt[X,Y]$.

\begin{lemma}\label{lem:gbqtopptop}
    A Gr\"obner basis $G$ for $\langle Q^\htop \cup P^\htop \rangle \subseteq \FFt[X,Y]$ is
    $$G = \left\{x_i \mid i=1,\ldots,n\right\} \cup \left\{ y_{i-1,j}y_{i,j}  \mid i=3,\ldots,n-1,\ j=1,\ldots,\ell_t \right\}\cup $$
    $$\left\{ y_{1,1}y_{2,1}\right\}\cup\left\{y_{1,j} \mid j=2,\ldots,\ell_t\right\}\cup\left\{y_{n,j} \mid j=1,\ldots,\ell_t\right\}\cup$$ 
    $$\left\{y_{i,j}^2 \mid i=2,\ldots,n-1,\ j=1,\ldots,\ell_t\right\} \cup \left\{y_{1,1}^2\right\}. $$
\end{lemma}
\begin{proof}
    The proof of this statements follows directly from inspecting Remark~\ref{rem:qtopptopdef} and the same observations as in proof of Lemma~\ref{lem:groebnerQh}.
    \qed
\end{proof}

The next theorem gives the exact value of the degree of regularity of the system $Q \cup P$. The proof uses similar arguments to those used for the proof of Theorem~\ref{Thm:Dreg-of-Qtop}.

\begin{theorem}\label{thm:dregQtopPtop}
    The degree of regularity of $Q \cup P$ is
    $$\dreg{Q \cup P} = \left\lceil\frac{n-1}{2} \right\rceil + (\ell_t - 1)\left\lceil\frac{n-2}{2} \right\rceil + 1.$$
\end{theorem}
\begin{proof}
    Define the set
    \[
        \tilde{Q}^\htop_{j,n} := Q^\htop_{j,n-1} \setminus \{y_{1,j}y_{2,j}\} \subset \FFt[y_{2,j},\ldots,y_{n-1,j}],
    \]
    for every $j=1,\ldots,\ell_t$.
    Let $G$ be a Gr\"obner basis of $\langle Q^\htop \cup P^\htop\rangle$ as in Lemma~\ref{lem:gbqtopptop}. Due to the presence of the linear monomials contributed by $P^\htop$, we observe 
    \begin{equation}\label{eq:qtopptopasunion}
    G = Q^\htop_{1,n-1} \cup \bigcup_{j=2}^{\ell_t} \tilde{Q}^\htop_{j,n-1} \cup \left\{x_i^2 \mid i=1,\ldots,n\right\}.
    \end{equation}
    Applying Corollary~\ref{cor:maximalmonomial}, we can get a monomial $m_1 \in \FFt[y_{1,1},\ldots,y_{n-1,1}]$ of maximal degree $\deg m_1 = \lceil (n-1)/2 \rceil$ such that $m_1 \not \in \FFt[y_{1,1},\ldots,y_{n-1,1}]/\langle Q^\htop_{1,n-1} \rangle$. We can obtain other $\ell_t - 1$ monomials $m_j$ of maximal degree $d = \lceil (n-2)/2 \rceil$, such that $m_j \not \in \langle \tilde{Q}^\htop_{h,n-1} \rangle$ for every $h=1,\ldots,\ell_t$ and every $j = 2,\ldots,\ell_t$.
    Let now $$m := \prod_{j=1}^{\ell_t}m_j \in \FFt[X,Y]/\langle G\rangle$$
    then $$d:=\deg m = \left\lceil\frac{n-1}{2} \right\rceil + (\ell_t - 1)\left\lceil\frac{n-2}{2} \right\rceil,$$ meaning that the $(d+1)$-th coefficient of the Hilbert series of  $\FFt[X,Y]/\langle G \rangle$ is $0$. The result on the degree of regularity $\dreg{Q\cup P}$ follows.
    \qed
\end{proof}

\begin{example}
    Let $n=8$ and $\ell_t=3$. 
    According to Theorem~\ref{thm:dregQtopPtop} the degree of regularity of $ Q \cup P$ is $11$.
%    $$\dreg{Q\cup P}=\left\lceil\frac{7}{2} \right\rceil + (3 - 1)\left\lceil\frac{6}{2} \right\rceil + 1 = 11.$$
     \noindent Using MAGMA, we compute and report the Hilbert series of the quotient  ring $\FFt[X,Y]/\langle Q^\htop \cup P^\htop\rangle$, i.e.
    \begin{align*}
        \mathrm{HS}_{\FFt[X,Y]/\langle Q^\htop \cup P^\htop\rangle} (z) =&\ 16z^{10} + 240z^9 + 1188z^8 + 2920z^7 + 4132z^6 + 3608z^5 +\\
        &\ 2005z^4 + 710z^3 + 155z^2 +
    19z + 1,
    \end{align*}
   thus $\dreg{Q\cup P}=\deg  \mathrm{HS}_{\FFt[X,Y]/\langle Q^\htop \cup P^\htop\rangle}+1=11$, matching our results.
\end{example}

\begin{remark}\label{rem:range fordregS}
    Since Theorem~\ref{thm:dregQtopPtop} considers the set $Q^\htop \cup P^\htop = S^\htop \setminus L_\vH^\htop$, it only gives an upper bound to the degree of regularity of $S$, that is $$\dreg{S}\leq \left\lceil\frac{n-1}{2} \right\rceil + (\ell_t - 1)\left\lceil\frac{n-2}{2} \right\rceil + 1.$$
 
\end{remark}
In the next section, we provide some experimental data showing the gap between the value computed in Theorem~\ref{thm:dregQtopPtop} and that of the actual solving degree. 

\subsection{Experimental results}
We performed several experiments for Modeling~\ref{modeling: improvedESD_F2} taking as input both random and Goppa codes, and we obtained a solving degree which is much smaller than the upper bound for the degree of regularity computed in Theorem~\ref{thm:dregQtopPtop}. This results in a much lower complexity estimate. We provide a selection of our experiments in Table~\ref{Tab:q2-SolveDeg}. The MAGMA code used for our experiments can be found \href{https://github.com/rexos/phd-cryptography-code/blob/main/modelings/Modeling2.txt}{here}.
\vspace{-2mm}

\begin{table}[H]
\centering
\begin{tabular}{|c|c|c||c|c|c|c|c|c|c|c|}
\hline 
\cellcolor{gray!20!}$n$& \cellcolor{gray!20!}$k$&\cellcolor{gray!20!}$t$& \cellcolor{gray!20!}Code Type&\cellcolor{gray!20!} \# Eqs &\cellcolor{gray!20!} \# Vars& \cellcolor{gray!20!}$d_{\mathrm{reg}}$ &\cellcolor{gray!20!} SR $d_{\mathrm{reg}}$&\cellcolor{gray!20!} $d_{\mathrm{M}}$&\cellcolor{gray!20!}Prange&\cellcolor{gray!20!} Modeling~\ref{modeling: improvedESD_F2} \\
\hline 
8      & 2      & 2    & Goppa        & (17,38)       & 24      & $\leq8$      & 3      & 2   & 9 & 23   \\
\hline 
%8      & 4      & 1    & Random       & 13      & 23       & 16      & [1,5]     & 2      & 3        \\ \hline 
10     & 5      & 4    & Random*    & (20,67)       & 40      & $\leq14$     & 5      & 3    & 10 & 38    \\ \hline 
%16     & 2      &  4    & Random       & 33 - 35     & 78 - 109       & 64     & [9,2]      & 3      & 3 -2         \\ \hline 
%16     & 4      & 2     & Random       & 31      & 78       & 48      & [4,16]    & 3      & 3 - 2       \\ \hline 
%16     & 8      & 1    & Random       & 25      & 47       & 32      & [1,9]    & 2      & 3 - 2      \\ \hline 
16     & 8      & 2    & Goppa        & (27,78)       & 48      & $\leq16$    & 5      & 3    & 12 & 40    \\ \hline 
20     & 10     & 5    & Random*    & (35,137)     & 80      & $\leq29$     & 7     & 3    & 12 & 46    \\ \hline 
30     & 15     & 7    & Random*    & (50,207)      & 120     & $\leq44$     & 10      & 4    & 15 & 65    \\ \hline 
32     & 12     & 4    & Goppa        & (57,221)      & 128     & $\leq 47$   & 10      &  3   & 16 &  52    \\ \hline 
%32     & 16     & 1    & Random       & 49      & 95       & 64      & [1,17]    & 3      & 4  - 3     \\ \hline 
%32     & 16     & 2    & Random       & 51      & 158      & 96      & [16,32]    & 5      & 4 - 3      \\ \hline 
32     & 17     & 3    & Goppa        & (50,158)      & 96      & $\leq32$   & 5      & 4    & 16 & 61   \\ \hline 
%32     & 20     & 1    & Random       & 45      & 95       & 64      & [5,17]   & 3      & 3        \\ \hline 
32     & 22     & 2    & Goppa        & (45,158)      & 96      & $\leq 32$   & 7     & 3     & 15 & 48  \\ \hline 
40 & 20 & 8 & Random* & (67,356) & 160 &  $\leq78$ & 16 & 5 & 17 & 88 \\ \hline
50 & 30 & 5 & Random & (75,347) & 200 & $\leq74$ & 15& 4 & 20 & 73\\ \hline
50 & 40 & 4 & Random & (65,347) & 200 & $\leq74$ & 17& 4 & 14 & 73\\ \hline
64 & 52 & 2 & Goppa & (79,318) & 192 & $\leq64$ & 14 & 4 & 18 & 72\\ \hline
64 & 40 & 4 & Goppa & (93,445) & 256 & $\leq95$ & 19 & 4 & 21 & 77\\ \hline
64 & 16 & 8 & Goppa & (119,572) & 320 & $\leq126$ & 21 & 4 & 19 & 81\\ \hline
\end{tabular}
\vspace{2mm}
\caption{This table shows experimental results for $\mathbb{F}_2$-linear codes using Modeling~\ref{modeling: improvedESD_F2} for the ESDP. 
The number of equations is given in the format (\#linear equations, \#quadratic equations). The values in the $d_{\mathrm{M}}$ column represent the smallest degree D such that MAGMA function GroebnerBasis(F,D) gives the Gr\"obner basis, i.e. highest step degree achieved when directly computing the Gröbner basis of the system in MAGMA, but ignoring the unnecessary steps in high degree that MAGMA F4 algorithm may do to insure termination. %This is typically regarded as a proxy for the solving degree $d_{\mathrm{sol}}$. 
The SR $d_{reg}$ column gives the degree of regularity of a semi-regular system of equations with the associated number of linear equations, quadratic equations, and variables, using \cite[Corollary 3.3.8]{B04}. The values in $d_{\mathrm{reg}}$ column are upper bounds for the degree of regularity of the system as provided by Theorem~\ref{thm:dregQtopPtop} and Remark~\ref{rem:range fordregS}. Random codes with ``*'' are decoding challenges from \url{https://decodingchallenge.org/syndrome}, with a number of errors slightly above Gilbert-Varshamov distance. %Several solutions are indeed found. 
The other random code instances are below GV distance instead. Instances with ``Goppa'' Code Type are random full-length binary Goppa codes with a number of errors equal to the Goppa polynomial degree. 
The Prange and Modeling 2 columns state the $\log_2$ of the complexities of the Prange algorithm (Esser-Bellini estimator~\cite{PKC:EssBel22}) and Gr\"obner basis computations \eqref{eq:GBcomplexity} with $d_{\mathrm{M}}$, respectively. %These parameters have also been tested with random codes and always provide the same solving degree in the two cases. %\rocco{which formula for SR degree of regularity are we using? The one for booelan system which intrisically takes into account field equations or the one foe generic fields? --had been using plus}
} \label{Tab:q2-SolveDeg}
\end{table}\vspace{-.2in}
\subsubsection{Comparison with combinatorial attacks.}
We also compare the complexity of solving ESDP using Gr\"obner basis techniques with the more traditional combinatorial attack proposed by Prange. The latter is widely regarded as a reference algorithm for decoding linear codes in the Hamming metric, and forms the basis of many state-of-the-art attacks on code-based cryptosystems.

However, Gröbner basis computations are known to be computationally expensive, particularly as the number of variables and the degree of the polynomials grow. In contrast, combinatorial methods like Prange exploit structured randomization and search to solve the problem more efficiently, albeit with large memory requirements.

In Table~\ref{Tab:q2-SolveDeg}, the complexity of the Prange algorithm has been computed using the Esser-Bellini cryptographic estimator~\cite{PKC:EssBel22}, see also this~\href{https://estimators.crypto.tii.ae/configuration?id=SDEstimator}{link}.
On the other hand, the complexity of solving the polynomial system of Modeling~\ref{modeling: improvedESD_F2} with Gr\"obner bases methods has been computed using formula~\eqref{eq:GBcomplexity} with $N = n(\lfloor\log_2(t)\rfloor + 2)$, the highest step degree $d_M$ achieved in the MAGMA experiments as a proxy for the solving degree $d_{\mathrm{sol}}$ and the conservative matrix multiplication constant $\omega = 2.807$.

These results confirm that direct combinatorial attacks like Prange outperform algebraic methods when solving the syndrome decoding problem, particularly for parameters of practical interest in code-based cryptography. Despite this, algebraic approaches remain valuable from a theoretical perspective and may offer insights into alternative solution strategies that could be leveraged in other contexts. We include this comparison to emphasize that the purpose of our Gröbner basis approach is not to compete with combinatorial attacks in efficiency, but rather to provide an alternative algebraic perspective on the syndrome decoding problem. Such perspectives can contribute to understanding the hardness of related problems in post-quantum cryptography.

%%%%%%%%%%%%%%%%%%%%%%%%%%%%%%%%%%%
%%%%%%%%%%%%%%%%%%%%%%%%%%%%%%%%%%%

\section{Modelings over $\FF_q$} \label{sec:Fq}

Each of the modelings we have discussed thus far (MPS, Modeling \ref{modeling: improvedSD_F2}, and Modeling \ref{modeling: improvedESD_F2})
%Both MPS and the improvements of Section~\ref{sec:EWM} 
are limited to the binary case. To the best of our knowledge, there is no modeling of the general syndrome decoding problem over $\FF_q$ for $q>2$ in the literature. In this section, we adapt the previous modelings to a generic finite field $\FF_q$, for some prime power $q\ge 2$, and explain how to efficiently (i.e. in polynomial time) obtain a polynomial system encoding an instance of the Syndrome Decoding Problem.

We will adopt the following notation throughout this section.

\begin{notation}\label{FQnotation}
   Let $n\geq2$ and let  $\CC \subseteq \FF_q^n$ be a $[n,k,d]$-linear code having a parity check matrix $\HH \in \FF_q^{(n-k) \times n}$. The vector $\vs \in \FF_q^{n-k}$ denotes the syndrome and $0\le t \le \lfloor (d-1)/2 \rfloor$ is the target error weight. Let $r_1,r_2>0$ be two integers. We will work over the polynomial ring $\FF_{q^{r_1}}[X,Y,Z]$, where $X=(x_1,\dots,x_n)$, $Y=(Y_1,\dots,Y_n)$, $Y_j=(y_{j,1}, \dots, y_{j,r_2})$, and $Z=(z_1,\dots,z_n)$ are variables. 
   
\end{notation}

As in the previous sections, $\vx=(x_1,\dots,x_n)$ is the vector of variables corresponding to the solution of the syndrome decoding problem. On the other hand, the role of the integers $r_1,r_2$ and of the variables $Y$ and $Z$ will be illustrated later.

We separately describe and explain the sets of polynomials that, together, model Problems~\ref{BSDP} and \ref{EWSDP}. Then we provide an analysis of the correctness of our modelings.

\subsection{Construction of the Equations}

\subsubsection{Identifying the support of a vector of $\mathbb{F}_q^n$.}

Unlike the Boolean case, the value of an element in the support of $\vx$ is not uniquely determined when $q\ge 3$. In order to count the number of nonzero coordinates with algebraic equations, we first need to map all nonzero elements to a unique element of $\FFq^*$, say 1.
Thus, in addition to the $Y$'s variables encoding the partial Hamming weights, we introduce here another length-$n$ vector of variables $Z=(z_1,\dots,z_n)$, each of which can only assume two values, 0 or 1, depending on whether the corresponding $X$ coordinate is nonzero. First, we tackle the problem of describing the relation between $X$ and $Z$ through algebraic equations. We distinguish two cases, depending on the target version of the problem, and then prove the sets of polynomials correctly describe our target.
\begin{itemize}
 \item \textit{Support constraint encoding for Problem~\ref{BSDP}.} Compute the following set of $2n$ quadratic polynomials
     \begin{equation} \label{eq:sceBSD}
         \{ x_j(z_j-1) \mid j=1,\dots,n\} \cup \{ z_j^2-z_j \mid j=1,\dots,n\}.
     \end{equation}
     \item \textit{Support constraint encoding for Problem~\ref{EWSDP}.} Compute the following set of $n$ polynomials of degree $q-1$
     \begin{equation} \label{eq:sceEWSD}
         \{z_j-x_j^{q-1}\mid j=1,\dots,n\}.
     \end{equation}
\end{itemize}

In the first case, the condition $z_j=1$ if $x_j\ne 0$ is given from the first set of polynomials. Otherwise, the second set implies $z_j\in\{0,1\}$. Therefore, the support of $(z_1,\ldots,z_n)$ contains the support of $(x_1,\ldots,x_n)$ and thus $\wt((z_1,\ldots,z_n))\ge \wt((x_1,\ldots,x_n))$.
In the second case, in order for the corresponding equations to be satisfied, $z_j=1$ if and only if $x_j\ne 0$, and $z_j=0$ otherwise. Hence $\wt((z_1,\ldots,z_n))=\wt((x_1,\ldots,x_n))$. 
%\ryann{Therefore, the support of $\vz$ contains the support of $\vx$ and thus $\wt(\vz)\ge \wt(\vx)$.}

From a computational point of view, the support constraint encoding for Problem~\ref{EWSDP} has a strong limitation, that is the high degree of the polynomials. A Gr\"obner basis computation would need to reach at least degree $q-1$ before taking into account such polynomials, leading to infeasible calculations unless $q$ is very small. This is reminiscent of the problem of including field equations in modelings over large fields. Yet, this issue does not appear in the support constraint encoding for Problem~\ref{BSDP}: the polynomials have constant degree 2 regardless of the field size $q$, making a modeling for Problem~\ref{BSDP} more realistic and valuable for effective computations.

\subsubsection{Hamming weight computation encoding.}

A difficulty arising from a direct generalization to large fields of the previous approach is the update of the weight registers, i.e. of the vectors $\vy_i$'s. In order to overcome this limitation, we introduce a different strategy for encoding the partial weights. More precisely, we substitute their binary expansion with vectors from a linear recurring sequence over an extension of $\FF_q$. As we will see, this approach naturally allows for the choice of different trade-offs between the number of variables and finite field size.

We first recall some known results about (univariate) polynomials over finite fields, companion matrices, and linear recurring sequences. We mainly refer to \cite{LN94} for this part.
\begin{definition}[Companion matrix, Chapter 2, \S5 \cite{LN94}]
Let $f(x)=x^d+f_{d-1}x^{d-1}+\dots+f_0\in \FF_q[x]$ be a monic polynomial. Its companion matrix is
\begin{equation} \label{eq: companion}
\Cf = \begin{bmatrix}
0          & 0          & \cdots & 0 & -f_0\\
1          & 0          & \cdots & 0 & -f_1\\
0          & 1          & \cdots & 0 & -f_2\\
\vdots & \vdots& \ddots &\vdots &\vdots\\
0 & 0 & \cdots & 1 & -f_{d-1}
\end{bmatrix}.
\end{equation}
\end{definition}

It is well known that the equation $f(\Cf)=0$ is satisfied, hence, if $f$ is a monic irreducible polynomial over $\FF_q$, then its companion matrix $\Cf$ plays the role of a root of $f$. It follows that the elements of the extension field $\FF_{q^d}$ can be written, according to this representation, as polynomials in $\Cf$ of degree strictly less than $d$.

We also recall that the order of a non-constant polynomial $f$ with $f_0\neq 0$ is the least positive integer $e$ such that $f(x) \mid x^e-1$. A polynomial in $\FF_q[x]$ of degree $d$ is said primitive if it is monic, $f(0)\ne 0$, and $\ord(f)=q^d-1$. Such polynomials can be found from the factorization of $x^{q^d-1}-1$.

The theory of linear recurring sequences says that the sequence of vectors $\vy_0,\Cf \vy_0,\Cf^2 \vy_0,\dots$, for some nonzero $\vy_0$ and companion matrix with $f_0\ne 0$, is periodic with least period equal to the order of $f$, when the latter is irreducible (cf.  \cite[Theorem 6.28]{LN94}). Therefore, by choosing $f$ primitive, we obtain a sequence of vectors $\vy_0,\Cf \vy_0,\Cf^2 \vy_0,\dots$ of maximal order $q^d-1$. On the other hand, the choice of $\vy_0$ does not seem to affect any property of our modeling. Without loss of generality, from now on we fix the initial state vector
\begin{equation}
             \vy_0=\begin{pmatrix}
            1&0& \cdots &0
            \end{pmatrix}^\top.
\end{equation}

By tuning the values $r_1$ and $r_2$ we can choose the number of variables used for the Hamming weight computation encoding, at the cost of working over more or less large field extensions. More precisely,
take
\[
m:= \min \{i \in \mathbb{N} \mid q^i > \max(t, n-t)+1 \},
\]
and let $r_1,r_2$ be two positive integers such that $m\le r_1 r_2$.
Then, let $f\in \FF_{q^{r_1}}[x]$ be a primitive polynomial of degree $r_2$ and $\Cf \in  \FF_{q^{r_1}}^{r_2\times r_2}$ its companion matrix. For convenience sake, we will use the column vector notation for the $Y_j$'s blocks of variables, i.e. 
$$Y_j= \begin{pmatrix}
        y_{j,1}&
        \cdots&
        y_{j,r_2}
        \end{pmatrix}^\top
      .$$ 
The polynomial encoding the partial Hamming weight is the following.
\begin{itemize}
    \item \textit{Hamming weight computation encoding.} Compute the $n r_2$ affine bilinear (in $Y$ and $Z$) polynomials from the expansion of
        \begin{equation} \label{eq: hwceFqinit}
        Y_1-(1-z_1)\cdot \vy_0-z_1\cdot\Cf\cdot \vy_0
    \end{equation}
    and
    \begin{equation} \label{eq: hwceFq}
        \{Y_j-(1-z_j)\cdot Y_{j-1}-z_j\cdot\Cf\cdot Y_{j-1} \;,\qquad \mathrm{for}\;j\in\{2,\ldots,n\}\}.
    \end{equation} 
\end{itemize}

\begin{remark}
    Using this approach, $r_1$ determines the finite field over which we define the resulting polynomial system (and thus the Multivariate Quadratic Problem instance), while $r_2$ determines the number of variables required for the weight computation encoding (which is strictly linked to the computational complexity). Observe that, together, equations \eqref{eq: hwceFqinit} and \eqref{eq: hwceFq} correspond to $nr_2$ polynomial equations over $\FF_{q^{r_1}}$ in $nr_2+n$ variables.  In several cases, working with a small amount of variables (namely, $r_1$ large and $r_2$ small) is preferable, but there are some instances in which working over small finite fields with a large number of variables (i.e. $r_1$ small and $r_2$ large) is advantageous. 
\end{remark}

The next result shows that the Hamming weight of $\vz$ is correctly computed.
\begin{proposition}\label{prop: correctness}
Consider the system given by \eqref{eq: hwceFqinit} and \eqref{eq: hwceFq} over $\FF_{q^{r_1}}[Y, Z]$. Any solution  $(\vy, \vz)=(\vy_1,\dots,\vy_n,\vz)\in \FF_{q^{r_1}}^{r_2 n}\times \{0,1\}^n$ of the system satisfies
\[\vy_j=\Cf^{\wt(\pi_j(\vz))}\vy_0.\]
In particular, $\vy_n=\Cf^{\wt(\vz)}\vy_0.$
\end{proposition}
\begin{proof}
It follows directly from the hypotheses by an inductive argument. 
\\
The first step is considering $\vy_1:=(y_{1,1},\ldots,y_{1,r_2})^{\top}$, which by definition is
$$ 
\vy_1=(1-z_1)\vy_0
+ z_1 \Cf \vy_0=
\left\{
\begin{array}{rl}
\vy_0&\mathrm{if}\; z_1=0\\
\Cf \vy_0&\mathrm{if}\; z_1\neq 0
\end{array}\;,
\right.
$$
namely, $\vy_1=\Cf^{\wt(z_1)}\vy_0=\Cf^{\wt(\pi_1(\vz))}\vy_0$.
\\
For the inductive step, we consider $\vy_{j-1}:=(y_{j-1,1},\ldots,y_{j-1,r_2})^{\top}$ to be equal to $\Cf^{\wt(\pi_{j-1}(\vz))}\vy_0$, and we look at the definition of $\vy_j$. We have
$$ 
\vy_j=(1-z_j)\vy_{j-1}+z_j \Cf \vy_{j-1}=\left\{
\begin{array}{cl}
\vy_{j-1}&\mathrm{if}\; z_j=0\\
\Cf \vy_{j-1}&\mathrm{if}\; z_j\neq 0
\end{array}\;,
\right.
$$
and either way, we obtain
$$
\vy_j=\Cf^{z_j}\cdot \vy_{j-1}=\Cf^{z_j}\Cf^{\wt(\pi_{j-1}(\vz))}\vy_0=\Cf^{\wt(z_j)+\wt(\pi_{j-1}(\vz))}\vy_0=\Cf^{\wt(\pi_{j}(\vz))}\vy_0.$$
\qed
\end{proof}

\subsubsection{Weight constraint encoding.} As for the previous modelings, the weight constraint encoding simply ensures that the representation of the last partial Hamming weight coincides with the representation of the total Hamming weight.
\begin{itemize}
    \item \textit{Weight constraint encoding.} Compute the $r_2$ affine linear polynomials in $Y$ from the expansion of
    \begin{equation} \label{eq: wceFq}
       \vy_n-\Cf^{t} \vy_0 .
        \end{equation}
\end{itemize}

\begin{corollary} \label{cor: Z}
    Consider the system given by \eqref{eq: hwceFqinit}, \eqref{eq: hwceFq} and \eqref{eq: wceFq} over $\FF_{q^{r_1}}[Y, Z]$ and let $z_1,\ldots,z_n$ be either $0$ or $1$.
    Then
    \begin{enumerate}
        \item The number of solutions $(\vy, \vz)\in \FF_{q^{r_1}}^{r_2 n}\times \{0,1\}^n$ of the system is equal to the number of binary vectors of Hamming weight equal to $t$, i.e. $\binom{n}{t}$;
        \item If $(\bar{\vy},\bar{\vz})$ and $(\tilde{\vy},\tilde{\vz})$ are two distinct solutions, then $\bar{\vz}\ne \tilde{\vz}$.
    \end{enumerate}
\end{corollary}
\begin{proof}
    From Proposition~\ref{prop: correctness} and \eqref{eq: wceFq}, it follows that any solution $(\vy, \vz)$ satisfies
    \[
    \Cf^t\vy_0= \begin{pmatrix}
        y_{n,1}\\
        \vdots\\
        y_{n,r_2}
        \end{pmatrix}
      =\Cf^{\wt(\vz)}\vy_0,
    \]
    which implies $\ord(f) \mid  \lvert t-\wt{(\vz)}\rvert$.
    Since $f$ is primitive, we obtain
    \[
    \ord(f)=(q^{r_1})^{r_2} -1 \ge q^m - 1> \max(t,n-t).
    \]
    On the other hand, $0\le \wt(\vz)\le n$, hence  $\lvert t-\wt{(\vz)}\rvert \le \min(t,n-t)$. Therefore, the only way $\ord(f)$ can divide $\lvert t-\wt{(\vz)}\rvert$, is that $\lvert t-\wt{(\vz)}\rvert=0$, i.e. $\wt(\vz)=t$. 
    
    Substituting $\vz$ in \eqref{eq: hwceFqinit} uniquely determines all the values $y_{1,1},\dots,y_{1,r_2}$ by linear equations of the form $y_{1,i}=c_i$. Moreover, substituting $y_{j-1,1},\dots,y_{j-1,r_2}$ in \eqref{eq: hwceFq} recursively determines all the values $y_{j,1},\dots,y_{j,r_2}$ in a similar manner. Therefore, if $\bar{\vz}= \tilde{\vz}$ are the projections over the last $n$ coordinates of two solutions, then also $(\bar{\vy},\bar{\vz})=(\tilde{\vy},\tilde{\vz})$, which concludes the proof. \qed
\end{proof}

\subsubsection{Field equations.}
The field equations concern only the $X$ part of the variables. 
\begin{itemize}
    \item \textit{Field equations.} The equations are obtained from the $n$ polynomials
    \begin{equation} \label{eq: ffe_Fq}
        \{x_j^q-x_j \mid j =1,\dots,n\}.
    \end{equation}
\end{itemize}
Indeed, $\vz$ already lies over $\{0,1\}^n$ because of the support constraint equations (and \eqref{eq: ffe_Fq}), while \eqref{eq: hwceFqinit} and \eqref{eq: hwceFq} force $Y$ to lie over $\FF_{q^{r_1}}$.

\subsection{The Modelings}
We are finally ready to describe the algebraic systems over $\FF_{q^{r_1}}[X,Y,Z]$ for Problems~\ref{EWSDP} and \ref{BSDP} and prove their correctness.

\begin{modeling}[Modeling for the SDP  over $\FF_q$] \label{modeling: BSD_Fq}
   Given an instance $(\HH,\mathbf{s},t)$ of Problem~\ref{BSDP} over $\FF_q$, Modeling~\ref{modeling: BSD_Fq} is the union of the sets of polynomials \eqref{eq:pce}, \eqref{eq:sceBSD}, \eqref{eq: hwceFqinit},  \eqref{eq: hwceFq}, \eqref{eq: wceFq} and \eqref{eq: ffe_Fq}.
\end{modeling}

\begin{modeling}[Modeling for the ESDP  over $\FF_q$] \label{modeling: EWSD_Fq}
   Given an instance $(\HH,\mathbf{s},t)$ of Problem~\ref{EWSDP} over $\FF_q$, Modeling~\ref{modeling: EWSD_Fq} is the union of the sets of polynomials \eqref{eq:pce}, \eqref{eq:sceEWSD}, \eqref{eq: hwceFqinit},  \eqref{eq: hwceFq}, \eqref{eq: wceFq} and \eqref{eq: ffe_Fq}.
\end{modeling}

As already said, finite field equations cannot be efficiently taken into account when dealing with large fields. In the exact weight syndrome decoding modeling, the support constraint equations are high-degree as well, so the problem would persist. On the other hand, it becomes convenient to remove the field equations in the bounded syndrome decoding problem. This leads to a new quadratic modeling.

\begin{modeling}[Quadratic Modeling for the SDP  over $\FF_q$] \label{modeling: quadBSD_Fq}
   Given an instance $(\HH,\mathbf{s},t)$ of Problem~\ref{BSDP} over $\FF_q$, Modeling~\ref{modeling: quadBSD_Fq} is the union of the sets of polynomials \eqref{eq:pce}, \eqref{eq:sceBSD}, \eqref{eq: hwceFqinit},  \eqref{eq: hwceFq}, \eqref{eq: wceFq}.
\end{modeling}

In the next section we thoroughly investigate the effect of removing the field equations from Modeling~\ref{modeling: EWSD_Fq}. We find, that at least for the parameter choices interesting for cryptography, the solutions of our modeling without field equations still lie over $\FFq$ with high probability.

Table~\ref{table:Fq-model-sizes} provides the number of variables and equations for the three modelings over $\FF_q$.
\begin{table}[H]
    \centering
\begin{tabular}{|c|c|c|c|}
	\hline
	& \# Polynomials & \# Variables & Degree\\
 \hline
	Modeling~\ref{modeling: BSD_Fq}	& $4n-k+n r_2+r_2$ & $n(r_2 + 2)$ & $q$\\
			\hline
   	Modeling~\ref{modeling: EWSD_Fq}	& $3n-k+n r_2+r_2$ & $n(r_2 + 2)$ & $q$\\
			\hline
   	Modeling~\ref{modeling: quadBSD_Fq}	& $3n-k+n r_2+r_2$ & $n(r_2 + 2)$ & $2$\\
			\hline
	\end{tabular}
  \vspace{2mm}
    \caption{Number of equations, number of variables and maximum degree of the algebraic modelings over $\FF_{q^{r_1}}$.}
    \label{table:Fq-model-sizes}
\end{table}

\begin{remark}
    Since $r_2$ can be chosen as at most $m=\OO(\log_q (n))$, both the number of polynomials and variables are quasi-linear in the code-length $n$ in all the three modelings, namely they are $\OO(\log_2 (n))$. At the cost of defining the system over $\FF_{q^{r_1}}=\FF_{q^{m}}$, these quantities become linear in $n$, as the choice $r_2=1$ is possible.
\end{remark}

The modelings above capture exactly the corresponding syndrome decoding problem variants.
\begin{theorem} \label{prop: sol iff}
Given an instance $(\HH,\vs,t)$,
\begin{enumerate}
\item The vector $(\vx, \vy, \vz)$ is a solution of Modeling~\ref{modeling: BSD_Fq} if and only if  $\vx$ is a solution of Problem~\ref{BSDP} and $\vx\in\FF_q$;
\item The vector $(\vx, \vy, \vz)$ is a solution of Modeling~\ref{modeling: EWSD_Fq} if and only if  $\vx$ is a solution of Problem~\ref{EWSDP} and $\vx\in\FF_q$;
\item The vector $(\vx, \vy, \vz)$ is a solution of Modeling~\ref{modeling: quadBSD_Fq} if and only if  $\vx$ is a solution of Problem~\ref{BSDP} and $\vx\in\overline{\FF_q}$, where $\overline{\FF_q}$ denotes the algebraic closure of $\FF_q$.
\end{enumerate}
\end{theorem}
\begin{proof}
Since the parity-check equations \eqref{eq:pce} belong to all three modelings, it remains to prove the conditions on the weight of $\vx$ and to determine the base field over the vector can lie.\\
    \textit{Proof of 1.} Modeling~\ref{modeling: BSD_Fq} contains the field equations, therefore $\vx\in \FF_q^n$. It has already been proved that \eqref{eq:sceBSD} implies $\wt(\vx)\le\wt(\vz)$. By Corollary~\ref{cor: Z}, \eqref{eq: hwceFqinit},  \eqref{eq: hwceFq} and \eqref{eq: wceFq} imply that $\wt(\vz)=t$, hence $\wt(\vx)\le t$.\\
    \textit{Proof of 2.} Modeling~\ref{modeling: EWSD_Fq} contains the field equations, therefore $\vx\in \FF_q^n$. It has already been proved that \eqref{eq:sceEWSD} implies $\wt(\vx)=\wt(\vz)$. By Corollary~\ref{cor: Z}, \eqref{eq: hwceFqinit},  \eqref{eq: hwceFq} and \eqref{eq: wceFq} imply that $\wt(\vz)=t$, hence $\wt(\vx)= t$.\\
    \textit{Proof of 3.} The proof is analogous to the proof of \textit{1.}, with the only exception that Modeling~\ref{modeling: quadBSD_Fq} does not contain the field equations. Hence, the solutions of the system are all the vectors defined over the algebraic closure $\overline{\FF_q}$ that satisfy the parity-check equations and have weight at most $t$. \qed
\end{proof}

\subsection{The Dimension of the Variety Associated with Modeling~\ref{modeling: quadBSD_Fq}}\label{Sec:Dim_of_Var}
Unlike the modelings that include the field equations, Modeling~\ref{modeling: quadBSD_Fq} is not a priori associated with a zero-dimensional ideal. This represents the main drawback of Modeling~\ref{modeling: quadBSD_Fq} compared to Modeling~\ref{modeling: BSD_Fq}. The zero-dimensional property is desirable because it is necessary for defining the degree of regularity and for applying the FGLM algorithm to convert the $\mathsf{degrevlex}$ Gröbner basis into a $\mathsf{lex}$ basis. While it is possible to convert non-zero dimensional ideals using methods such as Gr\"obner walk or others, the process may not be as straightforward \cite{GBwalk1, GBwalk2}.

In this subsection, we analyze the dimension of the variety associated with the ideal corresponding to Modeling~\ref{modeling: quadBSD_Fq}. We will explore the conditions under which the variety is zero-dimensional, as well as the probability of this occurring. We begin with the following reduction.

\begin{proposition}\label{prop:reductiontoProblem}
    Let $\vx\in\FF_q^n$ be a vector which is a solution of Problem~\ref{BSDP} for a given instance  $(\HH,\vs,t)$. In other words, $\vx$ satisfies $\mathbf{H}\vx^\top=\mathbf{s}^\top$ and  $\wt(\vx)\leq t$. Then, there exist finitely many $(\vy,\vz)$ such that  $(\vx, \vy, \vz)$ is a solution of Modeling~\ref{modeling: quadBSD_Fq}.
\end{proposition}

\begin{proof}
Let us first consider a vector $\vx$ satisfying the parity-check equations and such that $\wt(\vx)=t$. Then, in Modeling~\ref{modeling: quadBSD_Fq}, the $z_i$ must detect exactly the support of $\vx$, and $\vx$ uniquely determines a solution $(\vx,\vy,\vz)$ of the system. If instead $\wt(\vx)=\bar{t}<t$, then there exist $t-\bar{t}$ indexes where the $Z$ variables can have value 1 while $\vx_i=0$. Thus, for any solution $\bar{\vx}\in\FF_q^n$ of the parity check matrix of weight $\bar{t}$, there exist
\[
\binom{n-\bar{t}}{t-\bar{t}}
\]
different solutions $(\bar{\vx}, \vy, \vz)$. 

A special case is given by the codeword finding problem, i.e. where the syndrome $\vs$ is the zero vector. Here the null vector is a solution of Problem~\ref{BSDP} and leads to $\binom{n}{t}$ solutions, thus likely increasing a lot the solving degree and the cost of a Gr\"obner basis computation. We can get rid of all these solutions by fixing one variable $z_i=1$, thus forcing any solution to have weight at least 1 and removing the null vector. If the target solution has weight $t$, then the guess has success with probability $t/n$. We will discuss this strategy in more detail at the end of this subsection. 

\qed
\end{proof}

Proposition~\ref{prop:reductiontoProblem} implies that each solution to the decoding problem (Problem~\ref{BSDP}) corresponds to a finite number of solutions for Modeling~\ref{modeling: quadBSD_Fq}. This allows us to conduct an analysis that is independent of the specific modeling, as long as it accurately encodes the decoding problem in the sense of Theorem~\ref{prop: sol iff}. Therefore, we will focus on the Krull dimension of the solution set of Problem~\ref{BSDP} and provide a probability estimate for this dimension being zero.

First, in the following remark we briefly collect the definition of Krull dimension and some important properties we will use in the sequel. Expanded details and proofs can be found e.g. in \cite[\S4, Chapter~9]{cox1997ideals} or other standard references in commutative algebra and algebraic geometry.

\begin{remark}[Krull dimension]\label{prop: Krull=dim}\label{prop: max_dim}
Let $\K$ be an algebraically closed field (we will apply the following definitions and results to $\K=\overline{\FF_q}$). An affine variety $V$ is the zero locus in $\K^m$ of a proper ideal $I$ of the polynomial ring $\K[x_1,\dots,x_m]$. We say that $V$ is irreducible if it is not possible to write $V=V_1\cup V_2$ where $V_1,V_2\subsetneq V$ are two proper subvarieties. Irreducibility of $V$ is equivalent to the corresponding ideal $I$ being prime.
The \emph{Krull dimension} or simply the dimension of a variety $V$ is defined as the maximal length $d$ of the chains $V_0\subsetneq V_1\subsetneq \cdots \subsetneq V_d$,
of distinct nonempty irreducible subvarieties of $V$. This is also equivalent to the supremum of the lengths of all chains of prime ideals containing the defining ideal $I$ of $V$.
For example, the Krull dimension of an affine linear space $\mathcal{L}$ generated by $a$ linearly independent affine linear polynomials $L_1,\dots,L_a$ is precisely $m-a$, that is its dimension as an affine space. This can be seen by completing the polynomials to a maximal linearly independent system of $m$ equations (in $m$ variables) $L_1,\dots,L_a,L_{a+1},\dots,L_m$ and then considering the following maximal chain of prime ideals
\[
\mathcal{L}=\langle L_1,\dots L_a\rangle\subsetneq\langle L_1,\dots L_a,L_{a+1}\rangle\subsetneq\cdots\subsetneq \langle L_1,\dots L_m\rangle.
\]
Notice that each ideal in this chain is prime, being generated by linearly independent polynomials of degree $1$.
Finally, we mention that, thanks to the Noetherian property of the polynomial ring, a variety $V$ can be written uniquely as the union of irreducible varieties, which are called the irreducible components of $V$. Thus, the dimension of $V$ coincides with the largest of the dimensions of its irreducible components  (see \cite[\S 4,Chapter 9, Corollary 9]{cox1997ideals}).
\end{remark}

Let $S\subseteq [n]$. Given a matrix $\mathbf{H}$ with $n$ columns and a vector $\vx \in \FF_q^n$, we denote by $\mathbf{H}_S$ the submatrix of $\mathbf{H}$ of columns indexed by $S$ and by $\vx_S \in \FF_q^{|S|}$ the vector obtained by deleting the coordinates corresponding to $[n]\setminus S$ from $\vx\in\FF_q^n$. On the contrary, let $\mathsf{pad}_S(\vx)\in\FF_q^n$ be the vector obtained from $\vx\in \FF_q^{|S|}$ by padding with 0's the positions corresponding to $[n]\setminus S$.

\begin{proposition}
    Let $\mathcal{C}$ be an $[n,k]$ code with parity-check matrix $\mathbf{H}\in \FF_q^{(n-k)\times n}$. Then the set of solutions of Problem~\ref{BSDP} with target weight $t$ and syndrome $\vs$ for the code $\mathcal{C}$ is the finite union of irreducible components, namely
    \[
    \bigcup_{S \subset [n], |S|=t} \{ \mathsf{pad}_S(\vx) \in \FF_q^n \mid \mathbf{H}_S \vx= \vs\}.
    \]
\end{proposition}
\begin{proof}
    For any $S$ of cardinality $t$, $\wt(\mathsf{pad}_S(\vx))=\wt(\vx)\le t$ and $\mathbf{H}\cdot \mathsf{pad}_S({\vx}) =\mathbf{H}_S \vx=\vs$. Thus, the set of solutions of the decoding problem contains $$  \bigcup_{S \subset [n], |S|=t} \{ \mathsf{pad}_S(\vx) \in \FF_q^n \mid \mathbf{H}_S \vx= \vs\}.$$ On the other hand, for any solution $\vx \in \FF_q^n$ to the decoding problem, let $S$ be a set of cardinality $t$ containing the support of $\vx$. Then, $\vx\in \{ \mathsf{pad}_S(\vx) \in \FF_q^n \mid \mathbf{H}_S \vx= \vs\}$. Finally, all the sets $ \{ \mathsf{pad}_S(\vx) \in \FF_q^n \mid \mathbf{H}_S \vx= \vs\}$ are irreducible being affine linear spaces. \qed
\end{proof} 

The next proposition characterizes the dimension of the irreducible components of the set of solutions of the decoding problem and the finite field extension over which solutions are defined.
\begin{proposition} \label{prop: variety}
    The Krull dimension of the solution set of Problem~\ref{BSDP} with target weight $t$ and syndrome $\vs$ for the linear code with parity-check matrix $\mathbf{H}$ is 
    \begin{equation} \label{eq: dimension_Ideal}
    t-\min\{\rk(\mathbf{H}_S) \mid S\subseteq [n], |S|=t, \rk((\mathbf{H}_S\mid \vs))=\rk(\mathbf{H}_S)\}. 
    \end{equation}
    Moreover, if the dimension is 0, then all solutions lie over $\FF_q$.
    \end{proposition}
    \begin{proof}
Let us fix the support $S$ of $t$ possible error positions. By Remark~\ref{prop: Krull=dim}, the Krull dimension of the irreducible components coincides with their dimensions as affine linear spaces. The case study of the set of solutions of
\[
\mathbf{H}_S \vx = \vs,
\]
seen as a variety, thus becomes the following:

\begin{itemize}
    \item if $\rk((\mathbf{H}_S\mid \vs))>\rk(\mathbf{H}_S) \Rightarrow$ the variety is empty;
    \item if $\rk((\mathbf{H}_S\mid \vs))=\rk(\mathbf{H}_S) \Rightarrow$ the variety has dimension $t-\rk(\mathbf{H}_S)$. In particular, if the variety has dimension 0, i.e. $\rk(\mathbf{H}_S)=t$, then it has a unique element, which belongs to $\FF_q$.
\end{itemize}

By Remark~\ref{prop: max_dim}, the dimension of the solutions set is obtained as the maximum dimension over all the irreducible components corresponding to some $S$ for which $\rk((\mathbf{H}_S\mid \vs))=\rk(\mathbf{H}_S)$:
\begin{align*}
&\max \{ t - \rk(\mathbf{H}_S) \mid S\subseteq [n], |S|=t, \rk((\mathbf{H}_S\mid \vs))=\rk(\mathbf{H}_S)\}\\=& t-\min\{\rk(\mathbf{H}_S) \mid S\subseteq [n], |S|=t, \rk((\mathbf{H}_S\mid \vs))=\rk(\mathbf{H}_S)\}. 
\end{align*}
Let us now consider the case of a zero-dimensional variety. Then for any choice of $S$, there is at most a solution and it must belong to $\FF_q^n$. Hence, all solutions belong to $\FF_q^n$. \qed
\end{proof}

\begin{corollary} \label{cor: dim_variety}
    The dimension of the variety associated with Modeling~\ref{modeling: quadBSD_Fq} is
        \begin{equation}
    t-\min\{\rk(\mathbf{H}_S) \mid S\subseteq [n], |S|=t, \rk((\mathbf{H}_S\mid \vs))=\rk(\mathbf{H}_S)\}. 
    \end{equation}
\end{corollary}
\begin{proof}
    It readily follows from Propositions \ref{prop: max_dim} and \ref{prop: variety} and the fact that each solution of Problem~\ref{BSDP} corresponds to a finite number of solutions of Modeling~\ref{modeling: quadBSD_Fq}, thus it does not increase the dimension of the variety. \qed
\end{proof}

For relevant and not too-small parameters, we usually have $t\ll n-k$. Assuming the weight distribution of a linear code follows closely the Bernoulli one, we can estimate the probability that the ideal is zero-dimensional, and thus, by exploiting the proof of Proposition \ref{prop: variety}, that all the solutions $\vx$ lie over $\FF_q$.

\begin{proposition} \label{prop: bound}
Let $\mathcal{C}$ be an $\mathbb{F}_q$-linear code and let $W_i(\mathcal{C})$ the number of codewords of weight exactly $i$ in $\mathcal{C}$. Then the probability that Modeling~\ref{modeling: quadBSD_Fq} provides a variety of strictly positive dimension when $t<n-k$ is upper bounded by
\[
\sum_{i=1}^t  W_i(\mathcal{C})\left(\frac{1}{q^{n-k-i+1}}+\binom{n-i}{t-i}\left(\frac{1}{q^{n-k-t+1}}-\frac{1}{q^{n-k-i+1}}\right)\right)
\]
for a randomly sampled syndrome. For the codeword finding problem, the same probability is upper bounded by
\[
\sum_{i=1}^t  W_i(\mathcal{C})\binom{n-i}{t-i}.
\]
\end{proposition}
\begin{proof}

It follows from Corollary~\ref{cor: dim_variety} that the variety has positive dimension if and only if there exists a set $S\subseteq [n]$, $|S|=t$, such that $\mathbf{H}_S$ is not full-rank and the syndrome $\vs$ belongs to the column space of $\mathbf{H}_S$. 

The parity-check matrix $\mathbf{H}$ has $i$ linearly dependent columns indexed by the set $S$ if and only if the corresponding code $\mathcal{C}$ has a codeword of weight $\le i$ with support contained in $S$. Hence any codeword of positive weight $i\le t$ with support $S'$ identifies a set of $\binom{n-i}{t-i}$ supersets $S\supseteq S'$. On the other hand, each set $S$ of $t$ dependent columns is associated with \textit{at least} one codeword of weight $\le t$, hence iterating over such codewords is enough to guarantee an upper bound.

Let $W_i(\mathcal{C})$ be the number of codewords in $\mathcal{C}$ of weight exactly $i$. By splitting the event $\{\vs \in \mathsf{ColSpace}(\mathbf{H}_{S})\}$ into the union of the two disjoint events $\{\vs \in \mathsf{ColSpace}(\mathbf{H}_{\supp(\vc)})\}$ and  $\{\vs \in \mathsf{ColSpace}(\mathbf{H}_{S})\setminus \mathsf{ColSpace}(\mathbf{H}_{\supp(\vc)})\}$ and using that $\rk(\mathbf{H}_{\supp(\vc)})\le i-1$ and $\rk(\mathbf{H}_{S})\le \rk(\mathbf{H}_{\supp(\vc)})+(t-i)$, we thus obtain an upper bound on the sought probability:

\begin{align*}
&\mathbb{P}\left(\bigcup_{\substack{S\subseteq [n] \\ |S|=t \land \rk(\mathbf{H}_S)<t}} \{\vs\in \mathsf{ColSpace}(\mathbf{H}_S)\}\right)\\
\le &\sum_{\substack{S\subseteq [n] \\ |S|=t \land \rk(\mathbf{H}_S)<t}}\mathbb{P}(\{\vs\in \mathsf{ColSpace}(\mathbf{H}_S)\})\\
\le & \sum_{\substack{\vc \in \mathcal{C} \\ \wt(\vc)\le t} } \left( \mathbb{P}(\{\vs\in \mathsf{ColSpace}(\mathbf{H}_{\supp(\vc)})\}) + \sum_{\substack{\supp(\vc)\subseteq S\subseteq [n] \\ |S|=t }} \mathbb{P}(\{\vs\in \mathsf{ColSpace}(\mathbf{H}_S)\setminus \mathsf{ColSpace}(\mathbf{H}_{\supp(\vc)})\})\right)\\
= & \sum_{i=1}^t W_i(\mathcal{C}) \left( \frac{q^{\dim_{\FF_q} \mathsf{ColSpace}(\mathbf{H}_{\supp(\vc)})}}{q^{n-k}}+ \binom{n-i}{t-i}\frac{q^{\dim_{\FF_q} \mathsf{ColSpace}(\mathbf{H}_{S})}-q^{\dim_{\FF_q} \mathsf{ColSpace}(\mathbf{H}_{\supp(\vc)})}}{q^{n-k}} \right)\\
= & \sum_{i=1}^t W_i(\mathcal{C}) \left( \frac{q^{\rk(\mathbf{H}_{\supp(\vc)})}}{q^{n-k}}+ \binom{n-i}{t-i}\frac{q^{\rk(\mathbf{H}_{S})}-q^{\rk(\mathbf{H}_{\supp(\vc)})}}{q^{n-k}} \right)\\
\le & \sum_{i=1}^t W_i(\mathcal{C}) \left( \frac{q^{\rk(\mathbf{H}_{\supp(\vc)})}}{q^{n-k}}+ \binom{n-i}{t-i}\frac{q^{\rk(\mathbf{H}_{\supp(\vc)})}}{q^{n-k}}(q^{t-i}-1) \right)\\
\le & \sum_{i=1}^t W_i(\mathcal{C}) \left( \frac{1}{q^{n-k-i+1}}+ \binom{n-i}{t-i} \left(\frac{1}{q^{n-k-t+1}}-\frac{1}{q^{n-k-i+1}}\right) \right).
\end{align*}
Finally, in the case of the codeword finding problem, i.e. if the syndrome is the zero vector, the condition on the positive dimension of the variety boils down to the existence of the set $S$, $|S|=t$, such that the rank $\mathbf{H}_S$ is defective, as the zero vector belongs to any linear subspace. Therefore, in this case, the calculations are simplified into:
\[
\mathbb{P}\left(\bigcup_{\substack{S\subseteq [n] \\ |S|=t \land \rk(\mathbf{H}_S)<t}} \{0\in \mathsf{ColSpace}(\mathbf{H}_S)\}\right)
\le\sum_{\substack{S\subseteq [n] \\ |S|=t \land \rk(\mathbf{H}_S)<t}} 1 \le \sum_{i=1}^t W_i(\mathcal{C}) \binom{n-i}{t-i} .
\]
\qed\end{proof}

\begin{remark}
    For random codes, the weight distribution follows closely the Bernoulli one, i.e. $W_i(\mathcal{C})\approx \frac{\binom{n}{i}(q-1)^i}{q^{n-k}}$. Under this assumption, the probability of having a zero-dimensional ideal for the decoding modeling with a random syndrome can be estimated as
\[
\frac{1}{q^{n-k}}\sum_{i=1}^t  \binom{n}{i}(q-1)^i\left(\frac{1}{q^{n-k-i+1}}+\binom{n-i}{t-i}\left(\frac{1}{q^{n-k-t+1}}-\frac{1}{q^{n-k-i+1}}\right)\right),
\]
while for the codeword finding problem as
\[
\frac{1}{q^{n-k}}\sum_{i=1}^t  \binom{n}{i}\binom{n-i}{t-i}(q-1)^i.
\]
We remark in particular that the bound is independent from the choice of $(r_1,r_2)$. Different admissible pairs provide of course different solutions, but the projections of the varieties with respect to the $x_i$'s variables are the same over the field closure, which is what determines the dimension of the associated ideals.
\end{remark}

In Appendix~\ref{app:Dim_Exp}, Tables \ref{table: bound_SD} and \ref{table: bound_CF}, we provide examples of such bounds for concrete parameters. In the case of syndrome decoding, these probabilities are very small even at Gilbert-Varshamov distance and the issue of having ideals of positive dimension is thus absolutely negligible for the purpose of cryptanalysis. On the opposite, the bound on the probability of the same event for the codeword finding problem becomes completely useless when approaching the Gilbert-Varshamov distance. Indeed, the trivial upper bound ``$\pr \le 1$'' entries from the two tables mean that the bound given by Proposition \ref{prop: bound} gives a number larger than 1.
A possible workaround for the described issue with the codeword finding version is to make use of hybrid methods. Indeed, it is enough to guess a number of nonzero positions equal or greater than the ideal dimension to decrease the latter to 0 with high probability. Recalling that the solution space is projective and one the value of one nonzero entry can be chosen arbitrarily, specializing $l$ coordinates has a success probability of $\frac{\binom{n-l}{t-l}}{\binom{n}{t} (q-1)^{l-1}}$. 
In cryptanalysis, however, it is usually assumed to know the minimal weight of a (nonzero) solution. This is because, if there exists a solution of weight smaller than the target, then the challenge is actually easier. 

A simple strategy to obtain a zero-dimensional ideal in this setting is thus the following.
If we suppose to know a lower bound $d'$ the minimum distance $d(C)$ of the code, then it means that any $d'-1$ columns of the parity-check matrix $\mathbf{H}$ are linearly independent. Therefore, Equation~\eqref{eq: dimension_Ideal} implies that the dimension of the ideal for the bounded weight modeling with target weight $d'$ is exactly 1 and thus it is enough to specialize one variable $x_i$ to any element in $\FF_q$ to obtain a zero-dimensional ideal.

%%%%%%%%%%%%%%%%%%%%%%%%%%%%%%%%%%%
%%%%%%%%%%%%%%%%%%%%%%%%%%%%%%%%%%%

\subsection{Experimental results.}
We solved the quadratic system associated with Modeling~\ref{modeling: quadBSD_Fq} for several random codes. We show in Table~\ref{Tab:qNEQ2-SolveDeg} that, similarly to the case over $\FFt$, the solving degree is surprisingly small. MAGMA code used for our experiments with Modeling~\ref{modeling: quadBSD_Fq} can be found at \href{https://github.com/rexos/phd-cryptography-code/blob/main/modelings/Modeling5.magma}{this link}.

\begin{table}[h]
\resizebox{\textwidth}{!}{%
\begin{tabular}{|c|c|c|c|c|c|c|c|c|c|c|c|c|c|c|c|c|c|c|}
\hline 
\cellcolor{gray!20!} & \cellcolor{gray!20!} &\cellcolor{gray!20!}  & \cellcolor{gray!20!}  & \multicolumn{5}{|c|}{\cellcolor{gray!20!}$q=7$} & \multicolumn{5}{|c|}{\cellcolor{gray!20!}$q=16,17$}& \multicolumn{5}{|c|}{\cellcolor{gray!20!}$q=127$}\\
\hline
% \cellcolor{gray!20!}$n$  & \cellcolor{gray!20!}$k$  &\cellcolor{gray!20!} $t$  & \cellcolor{gray!20!}\#lin  & \cellcolor{gray!20!}$r_1$  & \cellcolor{gray!20!}$r_2$ & \cellcolor{gray!20!}\#quad & \cellcolor{gray!20!}\#vars  & \cellcolor{gray!20!}$d_{\mathrm{Magma}}$  & \cellcolor{gray!20!}$r_1$  & \cellcolor{gray!20!}$r_2$ & \cellcolor{gray!20!}\#quad & \cellcolor{gray!20!}\#vars  & \cellcolor{gray!20!}$d_{\mathrm{Magma}}$  & \cellcolor{gray!20!}$r_1$  & \cellcolor{gray!20!}$r_2$ & \cellcolor{gray!20!}\#quad & \cellcolor{gray!20!}\#vars  & \cellcolor{gray!20!}$d_{\mathrm{Magma}}$ \\ 
\cellcolor{gray!20!}$n$  & \cellcolor{gray!20!}$k$  &\cellcolor{gray!20!} $t$  & \cellcolor{gray!20!}\#lin  & \cellcolor{gray!20!}$r_1$  & \cellcolor{gray!20!}$r_2$ & \cellcolor{gray!20!}\#quad & \cellcolor{gray!20!}\#vars  & \cellcolor{gray!20!}$d_{\mathrm{M}}$  & \cellcolor{gray!20!}$r_1$  & \cellcolor{gray!20!}$r_2$ & \cellcolor{gray!20!}\#quad & \cellcolor{gray!20!}\#vars  & \cellcolor{gray!20!}$d_{\mathrm{M}}$  & \cellcolor{gray!20!}$r_1$  & \cellcolor{gray!20!}$r_2$ & \cellcolor{gray!20!}\#quad & \cellcolor{gray!20!}\#vars  & \cellcolor{gray!20!}$d_{\mathrm{M}}$ \\ 
\hline
 10 & 5 & 2 & 5 & \multicolumn{5}{|c|}{$m=2$} & \multicolumn{5}{|c|}{$m=1$} & \multicolumn{5}{|c|}{$m=1$} \\
 \cline{5-19}
 &&&& 2 & 1 & 30& 30& 4  & 1 & 1 & 30 & 30 & 4 & 1  & 1 & 30 & 30 & 4 \\
 &&&& 1 & 2 & 40& 40& 3  &&&&& &&&&& \\ \hline
 15 & 9 & 3 & 6 & \multicolumn{5}{|c|}{$m=2$} & \multicolumn{5}{|c|}{$m=1$} & \multicolumn{5}{|c|}{$m=1$} \\
 \cline{5-19}
 &&&& 2 & 1 & 45& 45& 4 &  1 & 1 & 45 & 45 & 4  &  1 & 1 & 45 & 45 & 4  \\
 &&&& 1 & 2 & 60& 60& 4  &&&&& &&&&& \\ \hline
 19 & 10 & 5 & 9 & \multicolumn{5}{|c|}{$m=2$} & \multicolumn{5}{|c|}{$m=1$} & \multicolumn{5}{|c|}{$m=1$} \\
 \cline{5-19}
 &&&& 2 & 1 & 57& 57& 5  & 1 & 1 & 57 & 57 & 5 &  1 & 1 & 57 & 57 & 5 \\
 &&&& 1 & 2 & 76& 76& 4  &&&&& &&&&& \\ \hline
 22 & 14 & 4 & 8 & \multicolumn{5}{|c|}{$m=2$} & \multicolumn{5}{|c|}{$m=2$} & \multicolumn{5}{|c|}{$m=1$} \\
 \cline{5-19}
 &&&& 2 & 1 & 66& 66& 5 &  2 & 1 & 66& 66& 5 &  1 & 1 & 66 & 66 & 4 \\
 &&&& 1 & 2 & 88& 88 & 4  & 1 & 2 & 88& 88 & 4  &&&&& \\ \hline
 30 & 20 & 4 & 10 & \multicolumn{5}{|c|}{$m=3$} & \multicolumn{5}{|c|}{$m=2$} & \multicolumn{5}{|c|}{$m=1$} \\
 \cline{5-19}
 &&&& 3 & 1 & 90& 90& $\ge 6$  & 2 & 1 & 90& 90& $\ge 6$   &  1 & 1 & 90 & 90 & $\ge 6$ \\
 &&&& 2 & 2 & 120& 120 & 4  & 1 & 2 & 120& 120 & 5  &&&&& \\
 &&&& 1 & 3 & 150& 150& 4  &&&&&& &&&&\\ \hline
\end{tabular}
}
\vspace{2mm}
\caption{This table gives information from experiments using random $\FF_q$-linear codes using Modeling~\ref{modeling: quadBSD_Fq}.
The values in the %$d_{\mathrm{Magma}}$ 
$d_{\mathrm{M}}$ 
column represent the highest step degree achieved when directly computing the Gröbner basis of the system in MAGMA. %This is typically regarded as a proxy for the solving degree $d_{\mathrm{sol}}$. 
The column ``\#lin'' denotes the number of linear equations, i.e. of parity-check equations, which is independent from the field size. The columns ``\#quad'' and ``\#vars'' stand for the number of quadratic equations and the number of variables, which depend on the value $r_2$ instead. The integer $r_1$ is the extension field degree over which the equations are defined. We recall that the value $m$ leads to different possible choices of $(r_1,r_2)$ and we give all minima with respect to the standard partial order on pairs.}
\label{Tab:qNEQ2-SolveDeg}
\end{table}

   %--------------------------
   \section{Conclusion and Future Directions} 

We have presented a new algebraic cryptanalysis for both the bounded and the exact versions of the Syndrome Decoding problem. 

In the binary case, our modelings significantly improved the previous attempt of \cite{2021/meneghetti}, by capturing the weight condition on the solution vector with quadratic polynomials. We have also experimentally shown that the behavior of the associated Gr\"obner basis is very different from that of a random system with the same number of variables and equations, \textit{leading to a much better complexity}. We have thus taken an important step towards making algebraic algorithms potentially competitive for the decoding problem.

We introduced algebraic modelings for the first time in the case of the general syndrome decoding problem over larger finite fields. Notably, one of them is quadratic with a number of variables and equations that is linear or quasi-linear in the code length, \textit{independently from the field size}. We have analyzed that, despite the constant degree of the equations involved, the system correctly solves the decoding problem and with high probability does not have spurious solutions for all parameters that are relevant to the problem.

\vskip 0.5cm
An open question to this work is to understand more clearly the behavior of the Gr\"obner basis computation both in the binary and in the general finite field cases and to get a theoretical estimate of the complexity that better matches with the one obtained from the experiments. This is a difficult task, as it is often the case for very structured algebraic systems, and probably requires to develop dedicated tools to analyze such behavior.

Another interesting and natural follow-up to this work can be to analyze the impact of hybrid strategies on solving the proposed systems. Since the weight of the solution sought is relatively low, a convenient choice is to set most of the variables of $X$ to 0. It is not difficult to see that this approach is reminiscent of the guess part in Prange or later ISD algorithms.

In the case of binary systems, we have verified experimentally that the best hybrid trade-off actually boils down to the Prange algorithm, the best complexity being indeed obtained when enough zeros to linearize the system are guessed.

However, the system hides a lot of structure and offers many different ways to specialize variables. For example, the auxiliary variables from the vector $\vy$ can also be fixed and they too have different probabilities of taking a value equal to 0 or 1. It is therefore not at all unrealistic to speculate that an ad-hoc and smart hybridization technique may lead to a better trade-off than a fully combinatorial approach.

\subsection*{Acknowledgments.}

We would like to thank the reviewers for their detailed and valuable feedback. Additionally, special thanks to Magali Bardet, Tanja Lange and Alberto Ravagnani for the fruitful discussions and insights.

This publication was created with the co-financing of the European Union FSE-REACT-EU, PON Research and Innovation 2014-2020 DM1062/2021. A. Caminata is supported by the PRIN 2020 grant 2020355B8Y ``Squarefree Gr\"obner degenerations, special varieties and related topics'', by the PRIN PNRR 2022 grant P2022J4HRR ``Mathematical Primitives for Post Quantum Digital Signatures'', by the MUR Excellence Department Project awarded to Dipartimento di Matematica, Università di Genova, CUP D33C23001110001, and by the European Union within the program NextGenerationEU. 
A. Meneghetti acknowledges support from Ripple's University Blockchain Research Initiative.
A. Caminata and A. Meneghetti are members of the INdAM Research Group GNSAGA.

    \bibliography{biblio.bib,crypto,abbrev0}
    \bibliographystyle{splncs04}

\appendix
\section{Section~\ref{Sec:Dim_of_Var} Bounds on the zero-dimensionality}\label{app:Dim_Exp}

%\begin{table}[ht!]
\begin{table}[H]
\centering
\resizebox{0.9\textwidth}{!}{%
\begin{tabular}{|c ||c c c c||} 
\hline
$[n,k]_q$& $t=\lfloor (d_{GV}-1)/2\rfloor$ & $t=\lfloor (n-k)/2\rfloor$ & $t=d_{GV}$ & $t=d_{GV}+1$ \\ \hline
$[100,50]_{2}$ & \makecell{$t=5,$ \\ $\pr\le  2.32\cdot 10^{-20} $} &  \makecell{$t= 25,$ \\ $\pr\le  1 $} &  \makecell{$t=12 ,$ \\ $\pr\le  6.73\cdot 10^{-9} $} &  \makecell{$t=13 ,$ \\ $\pr\le  1.84\cdot 10^{-7} $} \\ \hline
$[100,50]_{7}$ &  \makecell{$t= 12 , $ \\ $ \pr\le 8.44\cdot 10^{-51} $} &  \makecell{$t= 25 , $ \\ $ \pr\le  1.89 \cdot 10^{-20} $} &  \makecell{$t= 25 , $ \\ $ \pr\le  1.89 \cdot 10^{-20} $} &  \makecell{$t= 26 , $ \\ $ \pr\le  2.68 \cdot 10^{-18}$}  \\ \hline
$[100,50]_{127}$ &  \makecell{$t= 18 , $ \\ $ \pr\le  5.48 \cdot 10^{-118} $} &  \makecell{$t= 25 , $ \\ $ \pr\le  1.23 \cdot 10^{-84} $} &  \makecell{$t= 37 , $ \\ $ \pr\le  5.38 \cdot 10^{-30} $} &  \makecell{$t= 38 , $ \\ $ \pr\le  1.44 \cdot 10^{-25} $}  \\ \hline
$[100,80]_{2}$ &  \makecell{$t= 1 , $ \\ $ \pr\le  9.09 \cdot 10^{-11} $} &  \makecell{$t= 10 , $ \\ $ \pr\le 1 $} &  \makecell{$t= 4 , $ \\ $ \pr\le  3.14\cdot 10^{-4} $} &  \makecell{$t= 5 , $ \\ $ \pr\le  0.0268$}  \\ \hline
$[100,80]_{7}$ &  \makecell{$t= 3 , $ \\ $ \pr\le  4.06 \cdot 10^{-25} $} &  \makecell{$t= 10 , $ \\ $ \pr\le  2.92 \cdot 10^{-5} $} &  \makecell{$t= 8 , $ \\ $ \pr\le  1.30 \cdot 10^{-10} $} &  \makecell{$t= 9 , $ \\ $ \pr\le  6.55 \cdot 10^{-8} $}  \\ \hline
$[100,80]_{127}$ &  \makecell{$t= 6 , $ \\ $ \pr\le  1.16 \cdot 10^{-52} $} &  \makecell{$t= 10 , $ \\ $ \pr\le  1.14 \cdot 10^{-31} $} &  \makecell{$t= 13 , $ \\ $ \pr\le  1.97 \cdot 10^{-16} $} &  \makecell{$t= 14 , $ \\ $ \pr\le  1.97 \cdot 10^{-11} $}  \\ \hline
\end{tabular}
}
\vspace{2mm}
\caption{Bound on the probability $\pr$ that the ideal associated with the system has a strictly positive dimension for the decoding problem with a randomly sampled syndrome.}
\label{table: bound_SD}
\end{table} \vspace{-.5in}

%\begin{table}[ht!]
\begin{table}[H]
\centering
\resizebox{0.8\textwidth}{!}{%
\begin{tabular}{|c ||c c c c c||} 
\hline
$[n,k]_q$& $t=\lfloor (d_{GV}-1)/2\rfloor$ & $t=\lfloor (n-k)/2\rfloor$ & $t=d_{GV}-2$ & $t=d_{GV}-1$  & $t=d_{GV}$\\ \hline
$[ 100 , 50 ]_{ 2 }$ & \makecell{$t= 5 ,$ \\ $ \pr\le  2.07\cdot 10^{-6} $} & \makecell{$t= 25 ,$ \\ $ \pr\le  1 $} & \makecell{$t= 10 ,$ \\ $ \pr\le  1 $} & \makecell{$t= 11 ,$ \\ $ \pr\le  1 $} & \makecell{$t= 12 ,$ \\ $ \pr\le 1 $} \\ \hline
$[ 100 , 50 ]_{ 7 }$ & 
\makecell{$t= 12 ,$ \\ $ \pr\le  8.08\cdot 10^{-18} $} & \makecell{$t= 25 ,$ \\ $ \pr\le  1 $} & \makecell{$t= 23 ,$ \\ $ \pr\le  0.378 $} & \makecell{$t= 24 ,$ \\ $ \pr\le  1 $} &
\makecell{$t= 25 ,$ \\ $ \pr\le  1 $} \\ \hline
$[ 100 , 50 ]_{ 127 }$ & \makecell{$t= 18 ,$ \\ $ \pr\le  1.46\cdot 10^{-48} $} & \makecell{$t= 25 ,$ \\ $ \pr\le  6.16\cdot 10^{-30} $} & \makecell{$t= 35 ,$ \\ $ \pr\le  3.04\cdot 10^{-5} $} & \makecell{$t= 36 ,$ \\ $ \pr\le  0.00696 $} &
\makecell{$t= 37 ,$ \\ $ \pr\le  1 $} \\ \hline
$[ 100 , 80 ]_{ 2 }$ & \makecell{$t= 1 ,$ \\ $ \pr\le  9.54\cdot 10^{-5} $} & \makecell{$t= 10 ,$ \\ $ \pr\le  1 $} & \makecell{$t= 2 ,$ \\ $ \pr\le  0.0142 $} & \makecell{$t= 3 ,$ \\ $ \pr\le  1 $} & \makecell{$t= 4 ,$ \\ $ \pr\le  1 $} \\ \hline

$[ 100 , 80 ]_{ 7 }$ & \makecell{$t= 3 ,$ \\ $ \pr\le  6.93\cdot 10^{-10} $} & \makecell{$t= 10 ,$ \\ $ \pr\le  1 $} & \makecell{$t= 6 ,$ \\ $ \pr\le  0.00176 $} & \makecell{$t= 7 ,$ \\ $ \pr\le  0.165 $} & 
\makecell{$t= 8 ,$ \\ $ \pr\le  1 $} \\ \hline
$[ 100 , 80 ]_{ 127 }$ & 
\makecell{$t= 6 ,$ \\ $ \pr\le  4.20\cdot 10^{-21} $} & \makecell{$t= 10 ,$ \\ $ \pr\le  1.59\cdot 10^{-8} $} & \makecell{$t= 11 ,$ \\ $ \pr\le  1.65\cdot 10^{-5} $} & \makecell{$t= 12 ,$ \\ $ \pr\le  0.0155$} & \makecell{$t= 13 ,$ \\ $ \pr\le  1 $} \\ \hline
\end{tabular}
}
\vspace{2mm}
\caption{Bound on the probability $\pr$ that the ideal associated with the system has strictly positive dimension for the codeword finding problem (i.e. with null syndrome).}
\label{table: bound_CF}
\end{table}

\end{document}